%% file: price-of-fairness.tex
\documentclass[11pt,letterpaper]{article}

\usepackage{amssymb,amsmath,amssymb,bbm,mathtools} 
\PassOptionsToPackage{hyphens}{url}

\usepackage{algorithmic,algorithm}
\usepackage{eqparbox,array}
\usepackage{nicefrac}
\usepackage{xspace}

\usepackage[margin=1in]{geometry}
\usepackage{times}

\usepackage{amsthm}
\usepackage{thmtools,thm-restate}

\newtheorem{lemma}{Lemma}

\newtheorem{corollary}{Corollary}

\theoremstyle{definition}
\newtheorem{definition}{Definition}

\newtheorem{example}{Example}

\usepackage[square,sort,semicolon,numbers]{natbib} 

\usepackage[usenames,svgnames,xcdraw,table]{xcolor}
\definecolor{DarkBlue}{rgb}{0.1,0.1,0.5}
\definecolor{DarkGreen}{rgb}{0.1,0.5,0.1}
\usepackage{hyperref}
\hypersetup{
     colorlinks   = true,
     linkcolor    = magenta, 
     urlcolor     = magenta, 
	 citecolor    = magenta 
}

\usepackage[capitalise,noabbrev]{cleveref}

\makeatletter
\renewcommand{\paragraph}{%
  \@startsection{paragraph}{4}%
  {\z@}{1.0ex \@plus 1ex \@minus .2ex}{-1em}%
  {\normalfont\normalsize\bfseries}%
}
\let\oldnl\nl
\newcommand{\nonl}{\renewcommand{\nl}{\let\nl\oldnl}}
\makeatother

\usepackage{enumerate}
\usepackage{enumitem}
\usepackage[T1]{fontenc}
\usepackage{bbm}
\usepackage{comment}

\newcount\Comments  
\usepackage{color}
\definecolor{darkgreen}{rgb}{0,0.5,0}
\definecolor{purple}{rgb}{1,0,1}
\newcommand{\kibitz}[2]{\ifnum\Comments=1\textcolor{#1}{#2}\fi}

\renewcommand{\hat}{\widehat}
\newcommand{\floor}[1]{\lfloor {#1} \rfloor}

\newcommand{\argmax}[1]{{\operatorname{arg}\operatorname{max}}_{#1}}

\newcommand{\set}[1]{{\left\{#1\right\}}}
\newcommand{\tuple}[1]{{\left\langle#1\right\rangle}}
\newcommand{\bbN}{\mathbb{N}}
\newcommand{\bbR}{\mathbb{R}}
\newcommand{\calA}{\mathcal{A}}
\newcommand{\calB}{\mathcal{B}}
\newcommand{\calW}{\mathcal{W}}
\newcommand{\calI}{\mathcal{I}}
\newcommand{\calP}{\mathcal{P}}
\newcommand{\calU}{\mathcal{U}}
\newcommand{\calQ}{\mathcal{Q}}
\newcommand{\calT}{\mathcal{T}}

\newcommand{\EFone}{\textsc{EF1}\xspace}
\newcommand{\MMS}{\textsc{MMS}\xspace}
\newcommand{\sw}{\operatorname{SW}}

\newcommand{\OPT}{\operatorname{OPT}}
\newcommand{\hmms}{$\nicefrac{1}{2}$-$\MMS$\xspace}

\newcommand{\Gmms}{\Gamma_{\MMS}}
\newcommand{\Gsingle}{\Gamma_{\text{single}}}
\newcommand{\Ghard}{\Gamma_{\text{hard}}}

\Crefname{ALC@unique}{Line}{Lines}
\newcommand{\algefone}{\textsc{Alg}-$\EFone$-\textsc{High}\xspace}
\newcommand{\algefoneabs}{\textsc{Alg}-$\EFone$-\textsc{Abs}\xspace}
\newcommand{\algmms}{\textsc{Alg}-$\MMS$-\textsc{High}\xspace}
\newcommand{\algmmsabs}{\textsc{Alg}-$\MMS$-\textsc{Abs}\xspace}

\title{\bfseries Optimal Bounds on the Price of Fairness for Indivisible Goods}

\author{
	Siddharth Barman\thanks{Indian Institute of Science. {\tt barman@iisc.ac.in}}
	\and
	Umang Bhaskar\thanks{Tata Institute of Fundamental Research. {\tt umang@tifr.res.in}}
	\and
	Nisarg Shah\thanks{University of Toronto. {\tt nisarg@cs.toronto.edu}}
}
\date{}

\begin{document}	

\maketitle

\begin{abstract}
In the allocation of resources to a set of agents, how do fairness guarantees impact the social welfare? A quantitative measure of this impact is the price of fairness, which measures the worst-case loss of social welfare due to fairness constraints. While initially studied for divisible goods, recent work on the price of fairness also studies the setting of indivisible goods.

In this paper, we resolve the price of two well-studied fairness notions for the allocation of indivisible goods: envy-freeness up to one good ($\EFone$), and approximate maximin share ($\MMS$). For both $\EFone$ and \hmms guarantees, we show, via different techniques, that the price of fairness is $O(\sqrt{n})$, where $n$ is the number of agents. From previous work, it follows that our bounds are tight. Our bounds are obtained via efficient algorithms. For \hmms, our bound holds for additive valuations, whereas for $\EFone$, our bound holds for the more general class of subadditive valuations. This resolves an open problem posed by Bei et al.~(2019). 
\end{abstract}

\input{intro}

\input{notation}

\input{efone}

\input{mms}

\input{discussion}

\paragraph*{Acknowledgements.} SB gratefully acknowledges the support of a Ramanujan Fellowship (SERB - {SB/S2/RJN-128/2015}) and a Pratiksha Trust Young Investigator Award. UB's research is generously supported the Department of Atomic Energy, Government of India (project no. RTI4001), a Ramanujan Fellowship (SERB - SB/S2/RJN-055/2015), and an  Early Career Research Award (SERB - ECR/2018/002766). NS was partially supported by an NSERC Discovery Grant.

\bibliographystyle{splncs04nat}

\bibliography{fairness,abb,ultimate}

\appendix
\section*{Appendix}

\input{lbounds}

\end{document}

%% file: intro.tex
\section{Introduction}\label{sec:intro}
What does it mean for an allocation of resources to a set of agents to be fair? 
The most compelling notion of fairness advocated in prior work is \emph{envy-freeness} (EF)~\cite{Fol67}, which demands that no agent envy another agent (i.e., value the resources allocated to the other agent more than those allocated to herself). When the resources to be allocated contain \emph{indivisible} goods, guaranteeing envy-freeness is impossible.\footnote{The canonical example is that of a single indivisible good and two agents; the agent who does not receive the good will inevitably envy the other.} Thus researchers have sought relaxations such as \emph{envy-freeness up to one good}~\cite{Bud11,CKMP+19}, which states that it should be possible to eliminate any envy one agent has toward another by removal of at most one good from the latter's allocation. 

A different relaxation of envy-freeness advocated in the context of indivisible goods is the \emph{maximin share guarantee} (\MMS)~\cite{Bud11}. When there are $n$ agents, the maximin share value of an agent is defined as the maximum value she could obtain if she were to divide the goods into $n$ bundles, and then receive the least-valued bundle according to her. An \MMS allocation is an allocation where each agent receives at least her maximin share value. This generalizes the classical cut-and-choose protocol for dividing a cake between two agents, where one agent cuts the cake into two pieces, the other agent chooses a piece, and the first agent then receives the remaining piece. Even for agents with additive valuations over the goods,\footnote{A valuation function is additive if the value of a bundle of goods is the sum of the values of the individual goods in the bundle.} \MMS allocations may not exist~\cite{PW14,KPW16}. However, an allocation where each agent receives at least $\nicefrac{3}{4}$-th of her maximin share value --- i.e., a $\nicefrac{3}{4}$-\MMS allocation --- always exists~\cite{garg2020improved,ghodsi2018fair}.

While fairness is important, it is clearly not the only objective of interest. A competing criterion is the aggregate value of the agents from the resources they receive, or the \emph{social welfare} of the allocation. The tradeoff between fairness and social welfare is quantitatively measured by the \emph{price of fairness}, the supremum ratio of the maximum welfare of any allocation to the maximum welfare of any allocation satisfying the desired fairness notion, where the supremum is over all possible problem instances. Intuitively, this measures the factor by which welfare may be lost to achieve the desired fairness. 

\citet{CKKK09} initiated the study of the price of fairness in the canonical setting of cake-cutting, in which a heterogeneous \emph{divisible} good is to be allocated (thus, envy-freeness can be guaranteed) and agents have additive valuations. They proved that the price of envy-freeness is between $\Omega(\sqrt{n})$ and $O(n)$, where $n$ is the number of agents. Later, \citet{BFT11} closed the gap by proving that the correct bound is $\Theta(\sqrt{n})$, and that the matching upper bound can be achieved by maximizing the Nash welfare~\cite{BFT11}.

For \emph{indivisible} goods and additive valuations, \citet{BeiLMS19} studied the price of fairness for various notions of fairness. They showed that the price of envy-freeness up to one good ($\EFone$) is between $\Omega(\sqrt{n})$ and $O(n)$. One might immediately wonder if maximizing the Nash welfare, which is known to satisfy EF1 when allocating indivisible goods~\cite{CKMP+19}, can be used to derive a matching $O(\sqrt{n})$ upper bound, like in cake-cutting. Unfortunately, \citeauthor{BeiLMS19} also showed that maximizing the Nash welfare results in a loss of welfare by $\Omega(n)$ factor in the worst case, and posed settling the price of $\EFone$ as a significant open question. For approximate $\MMS$ allocations (since exact $\MMS$ allocations may not exist), the price of fairness is not studied earlier, though a lower bound of $\Omega(\sqrt{n})$ can be obtained from previous constructions~\cite{CKKK09}. 

\paragraph{Our contributions.} We first describe our model more formally. We consider the allocation of $m$ indivisible goods to $n$ agents. An allocation $\calA$ partitions the goods into $n$ bundles, where $A_i$ is the bundle assigned to agent $i$. The preferences of agent $i$ are encoded by a valuation function $v_i$, which assigns a non-negative value to every subset of goods. Then, $v_i(A_i)$ is the value to agent $i$ under $\calA$, and $\sum_i v_i(A_i)$ is the social welfare of $\calA$. The price of a fairness notion, as defined above, is the supremum ratio of the maximum social welfare to the maximum social welfare subject to the fairness notion. 

Our main contribution is to comprehensively settle the price of \EFone and \hmms fairness notions. First, we show that the price of \EFone is $O(\sqrt{n})$. This matches the $\Omega(\sqrt{n})$ lower bound due to \citet{BeiLMS19}. The lower bound is for additive valuations, whereas our upper bound holds for the more general class of subadditive valuations. Hence our work settles this open question for all valuation classes between additive and subadditive. Our upper bound is obtained via a polynomial-time algorithm. For subadditive valuations, given the lack of succinct representation, we assume access to a \emph{demand-query oracle}.\footnote{A formal description of valuation classes and query models is in~\Cref{sec:prelim}.} As a consequence of this result, we also settle the price of a weaker fairness notion --- proportionality up to one good (Prop1) --- as $\Theta(\sqrt{n})$ for additive valuations. 

For the \hmms fairness notion and additive valuations, we similarly show, via a different algorithm, that the price of fairness is $\Theta(\sqrt{n})$. We show that for a fixed $\epsilon > 0$, a $(\nicefrac{1}{2}-\epsilon)$-$\MMS$ allocation with welfare within $O(\sqrt{n})$ factor of the optimal can be computed in polynomial time. 

\paragraph{Related work.} For resource allocation, the price of fairness was first studied by \citet{CKKK09} and \citet{BFT11}. While \citet{CKKK09} left open the question of the price of envy-freeness in cake-cutting, later settled by \citet{BFT11} as $\Theta(\sqrt{n})$, they proved that the price of a weaker fairness notion --- proportionality --- is $\Theta(\sqrt{n})$, and the price of an incomparable fairness notion --- equitability --- is $\Theta(n)$. They also extended their analysis to the case where the agents \emph{dislike} the cake (i.e., a divisible \emph{chore} is being allocated), and the case with indivisible goods or chores. However, in case of indivisible goods, notions such as envy-freeness, proportionality, and equitability cannot be guaranteed. The analysis of \citet{CKKK09} simply excluded instances which do not admit allocations meeting these criteria. 

\citet{BeiLMS19} instead focused on notions that can be guaranteed with indivisible goods and chores, such as envy-freeness up to one good (EF1). While they did not settle the price of EF1 (which is the focus of our work), they showed that the price of popular allocation rules such as the maximum Nash welfare rule~\cite{CKMP+19}, the egalitarian rule~\cite{Rawls71}, and the leximin rule~\cite{BM04,kurokawa2018leximin} is $\Theta(n)$. They also considered allocations that are balanced, i.e., give all agents an approximately equal number of goods, and settled the price of this guarantee as $\Theta(\sqrt{n})$. To the best of our knowledge, our work is the first to consider the price of approximate maximin share guarantee. Recent work has also analyzed the price of fairness in cake-cutting when the piece allocated to each agent is required to be contiguous~\cite{aumann2015efficiency}. \citet{suksompong2019fairly} extended this analysis to indivisible goods.

The price of fairness has also been analyzed in other paradigms at the intersection of computer science and economics, such as in kidney exchange~\cite{DPS14}, and is inspired from other similar worst-case measures such as the price of anarchy and of stability in game theory~\cite{koutsoupias1999worst}. More distantly, the price of fairness has been extensively analyzed in machine learning (see, e.g.,~\citet{mehrabi2019survey}).

%% file: notation.tex
\section{Preliminaries}\label{sec:prelim}

For $k \in \bbN$, define $[k] \coloneqq \set{1,\ldots,k}$. We study discrete fair division problems, wherein a set $[m]$ of indivisible goods need to be partitioned in a fair manner among a set $[n]$ of agents. 

\paragraph{Agent valuations.} The cardinal preference of each agent $i \in [n]$ (over the goods) is specified via the valuation function $v_i: 2^{[m]} \mapsto \bbR_{\ge 0}$, where $v_i(S) \in \bbR_{\ge 0}$ is the value that agent $i$ has for the subset of goods $S \subseteq [m]$.   We assume valuations are normalized ($v_i(\emptyset) = 0$), nonnegative ($v_i(S) \geq 0$ for all $S \subseteq [m]$), and monotone ($v_i(A) \leq v_i(B)$ for all $A \subseteq B \subseteq [m]$). A fair-division instance is represented by a tuple $\tuple{[n], [m], \{v_i\}_{i \in [n]}}$. We primarily consider two valuation classes:

\begin{itemize}
\item Additive: $v_i(S) = \sum_{g \in S} v_i(g)$ for each agent $i \in [n]$ and subset of goods $S \subseteq [m]$. Here, $v_i(g)$ denotes the value that agent $i$ has for good $g \in [m]$. 
\item Subadditive: $v_i(S \cup T) \leq v_i(S) + v_i(T)$ for each agent $i \in [n]$ and all subsets $S, T \subseteq [m]$. 
\end{itemize}
Note that the family of subadditive valuations encompasses additive valuations.

\paragraph{Oracle access.} Since describing subadditive valuations may require size exponential in the number of agents and goods, to design efficient algorithms, we assume oracle access to the valuation functions. The literature focuses on two prominent query models.
\begin{itemize}
	\item Value query: Given an agent $i \in [n]$ and a subset of goods $S \subseteq [m]$, the oracle returns $v_i(S)$. 
	\item Demand query: Given an agent $i \in [n]$ and a price $p_g \in \bbR_{\ge 0}$ for each good $g \in [m]$, the oracle returns a ``profit-maximizing'' set $S^* \in \argmax{S \subseteq [m]} v_i(S)-\sum_{g \in S} p_g$. 
\end{itemize}
Demand queries are \emph{strictly} more powerful than value queries (e.g.,~\cite[Section 11.5]{NisanRTV07}).

\paragraph{Allocations.} Write $\Pi_n([m])$ to denote the set of all $n$-partitions of the set of goods $[m]$. An allocation $\calA = (A_1, A_2, \ldots, A_n) \in \Pi_n([m])$ corresponds to an $n$-partition wherein the subset $A_i \subseteq [m]$ is assigned to agent $i \in [n]$; such a subset is called a bundle. The term partial allocation denotes an $n$-partition $(P_1, P_2, \ldots, P_n) \in \Pi_n(S)$ of a subset of goods $S \subseteq [m]$; as before, subset $P_i$ is assigned to agent $i \in [n]$. 

\paragraph{Fairness.} The notions of fairness considered in this work are defined next. 

\begin{definition}[Prop1] An allocation $\calA$ is called proportional up to one good (Prop1) if for each agent $i \in [n]$, there exists a good $g \in [m]$ such that $v_i(A_i \cup \set{g}) \ge v_i([m])/n$.
\end{definition}

\begin{definition}[EF1] An allocation $\calA$ is called envy-free up to one good (\EFone) iff for all agents $i, j \in [n]$ with $A_j \neq \emptyset$, there exists a good $g \in A_j$ such that $v_i(A_i) \geq v_i(A_j \setminus \set{g})$. 
\end{definition}

It is easy to check that EF1 implies Prop1 for additive valuations. In a fair-division instance $\calI = \tuple{[n], [m], \{v_i\}_{i \in [n]}}$, the maximin share of an agent $i \in [n]$ is defined as 
\begin{align*}
\MMS_i & \coloneqq \max_{(P_1, \ldots, P_n) \in \Pi_n([m])} \ \min_{j \in [n]} \ v_i (P_j)
\end{align*}

\begin{definition}[$\alpha$-$\MMS$] For $\alpha \in [0,1]$, an allocation $\calA$ is called $\alpha$-approximate maximin share fair ($\alpha$-\MMS) if $v_i(A_i) \geq \alpha \cdot \MMS_i$ for each agent $i \in [n]$.  
\end{definition}

\paragraph{Social welfare.} The social welfare of an allocation $\calA$, denoted $\sw(\calA)$, is defined as the sum of the values that the agents derive from the allocation: $\sw(\calA) = \sum_{i=1}^n v_i(A_i)$. We will use $\mathcal{W}^*=(W^*_1, W^*_2, \ldots, W^*_n)$ to denote a social welfare maximizing allocation, i.e., $\mathcal{W}^*$ $\in \argmax{\calA \in \Pi_n([m])} \sw(\calA)$, and $\OPT $ to denote the optimal social welfare, i.e., $\OPT = \sw(\mathcal{W}^*)$. 

\paragraph{Price of fairness.} Given a fairness property $X$, the price of $X$ is the supremum, over all fair division instances with $n$ agents and $m$ goods, of the ratio between the the maximum social welfare of any allocation to the maximum social welfare of any allocation satisfying property $X$. 

\paragraph{Scaling.} To ensure that agent valuations are on the same scale, much of the literature on fair division assumes that agent valuations are \emph{scaled}, i.e. $v_i([m]) = 1$ for all $i \in [n]$. Noting that unscaled valuations are common in other areas of social choice (e.g.,~\cite{DuttaPS07}), we consider both scaled and unscaled valuations. 

%% file: efone.tex
\section{Price of Envy-Freeness Up To One Good (EF1)}\label{sec:ef1}

We begin by studying the price of fairness for $\EFone$ allocations for agents with subadditive valuations. For scaled additive valuations, \citet{BeiLMS19} show that the price of $\EFone$ is between $\Omega(\sqrt{n})$and $O(n)$. We tighten their upper bound to $O(\sqrt{n})$ (thus matching their lower bound) even when the valuations are subadditive. For unscaled valuations, we show that the bound is $\Theta(n)$. Our main result in this section is as follows. 

\begin{restatable}{theorem}{TheoremEFone}
\label{theorem:efone}
The price of $\EFone$ is $O(\sqrt{n})$ for scaled subadditive valuations and $O(n)$ for unscaled subadditive valuations. Both bounds are tight even when the valuations are additive. 
\end{restatable}

We begin by proving the upper bounds, which are established using \Cref{alg:efone-half,alg:efone}. 

\subsection{An Absolute Welfare Guarantee}\label{sec:ef1-absolute}

First, we show that when agents have subadditive valuations $\set{v_i}_{i \in [n]}$ (not necessarily scaled), there always exists an $\EFone$ allocation $\calA$ with social welfare $\sw(\calA) \ge \frac{1}{2n} \sum_{i=1}^n v_i([m])$. 

This result has two implications. First, since $\sum_{i=1}^n v_i([m])$ is a trivial upper bound on the optimal social welfare $\OPT$, the result establishes an $O(n)$ upper bound on the price of $\EFone$. For unscaled valuations, this is exactly the bound we need. For scaled valuations, we need to improve this to $O(\sqrt{n})$. Since $\sum_{i=1}^n v_i([m]) = n$ under scaled valuations, the result gives $\sw(\calA) \ge \nicefrac{1}{2}$. Hence, if $\OPT = O(\sqrt{n})$, then we have the desired $O(\sqrt{n})$ upper bound. We analyse the case when $\OPT = \Omega(\sqrt{n})$ in \Cref{sec:ef1-high-opt}. Our absolute welfare guarantee is derived through \Cref{alg:efone-half}. 

\begin{algorithm}[htb]
  \caption{\algefoneabs} \label{alg:efone-half}  
  \textbf{Input:} Fair-division instance $\calI= \langle [n],[m],\set{v_i}_{i \in [n]} \rangle$ with value-query oracle access to the subadditive valuations $v_i$s.
  
  \textbf{Output:} An $\EFone$ allocation $\calB$ with social welfare $\sw(\calB) \ge \frac{1}{2n} \sum_{i \in [n]} v_i([m])$.
  
  \begin{algorithmic}[1]
  \STATE Consider the weighted bipartite graph $G = ([n] \cup [m], [n] \times [m])$ with weight of each edge $(i, g) \in [n] \times [m]$ set as $v_i(g)$. Let $\pi$ be a maximum-weight matching in $G$ that matches all nodes in $[n]$. \label{step:max-match}
  \STATE Construct the partial allocation $\calB'$ such that $B'_i = \set{\pi(i)}$ for each $i \in [n]$. 
  \COMMENT{Note that $\calB'$ is trivially $\EFone$ because each agent is assigned a single good.}
  \STATE Use the algorithm of \citet{LMMS04} to extend the partial $\EFone$ allocation $\calB'$ into a complete $\EFone$ allocation $\calB$ such that $v_i(B_i) \ge v_i(B'_i)$ for each agent $i \in [n]$. \label{step:Lipton-abs}
  \RETURN Allocation $\calB$
\end{algorithmic}
\end{algorithm}

\begin{lemma}
\label{lemma:efone-half}
Let $\calI = \langle [n], [m], \{v_i\}_{i \in [n]} \rangle$ be a fair-division instance in which agent valuations are subadditive. Then, given value-query oracle access to the valuations, \algefoneabs (\Cref{alg:efone-half}) efficiently computes an $\EFone$ allocation $\calB$ with social welfare $\sw(\calB) \ge \frac{1}{2n} \sum_{i \in [n]} v_i([m])$. 
\end{lemma}
\begin{proof}
For each agent $i \in [n]$, sort the goods as $g_{i_1},\ldots,g_{i_m}$ in a non-decreasing order of their value to agent $i$ (with ties broken arbitrarily), i.e., so that $v_i(g_{i_1}) \ge \ldots \ge v_i(g_{i_m})$, and let $G_i = \set{g_{i_1},\ldots,g_{i_n}}$ be the set of $n$ most valuable goods to agent $i$.

Consider a subgraph $G' = ([n] \cup [m], E')$ of the weighted bipartite graph $G$ constructed in \Cref{step:max-match} of \Cref{alg:efone-half}, in which we only retain the edges $E' = \set{(i,g) : i \in [n],g \in G_i}$ (i.e., from each agent $i$ to her $n$ most valuable goods). Because $G'$ is $n$-left-regular, it can be decomposed into $n$ disjoint left-perfect matchings $M_1,\ldots,M_n$; this is an easy consequence of Hall's theorem~\cite{Bol02}. Due to the pigeonhole principle, one of $M_1,\ldots,M_n$ must have weight at least $\frac{1}{n} \cdot \sum_{i \in [n], g \in G_i} v_i(g)$. Since this matching is also a left-perfect matching in $G$, and $\pi$ constructed in \Cref{step:max-match} is a maximum-weight left-perfect matching in $G$, we have $\sum_{i \in [n]} v_i(\pi(i)) = \sw(\calB') \ge \frac{1}{n} \cdot \sum_{i \in [n], g \in G_i} v_i(g)$. 

Note that $\calB'$ assigns a single good to each agent; hence, it is trivially $\EFone$. It is known that for any monotonic valuations (thus, in particular, for subadditive valuations), the algorithm of \citet{LMMS04} can extend a partial $\EFone$ allocation $\calB'$ into a complete $\EFone$ allocation $\calB$ such that $v_i(B_i) \ge v_i(B'_i)$. Thus, for the $\EFone$ allocation $\calB$ returned by the algorithm, we have 
\begin{align}\label{ineq:max-matching} 
\sw(\calB) \geq \frac{1}{n} \sum_{i \in [n]} \sum_{g \in G_i} v_i(g).
\end{align} 

Using the fact that $\calB$ is $\EFone$, we get that, for each $i,j \in [n]$, there exists $S_j \subseteq B_j$ with $|S_j| \le 1$ such that $v_i(B_i) \geq v_i(B_j \setminus S_j) \geq v_i(B_j) - v_i(S_j)$; here the last inequality follows from the subadditivity of $v_i$. Summing this inequality over $j \in [n]$, we get 
\begin{align*}
n \cdot v_i(B_i) \geq \sum_{j \in [n]} v_i(B_j) - \sum_{j \in [n]} v_i(S_j) \geq v_i([m]) -  \sum_{j \in [n]} v_i(S_j) \geq v_i([m]) - \sum_{g \in G_i} v_i(g),
\end{align*}
where the penultimate inequality holds because $B_j$s form a partition of $[m]$ and $v_i$ is subadditive, and the last inequality holds because the $n$ sets $S_j$s are disjoint and have cardinality at most $1$, and $G_i$ is the set of the $n$ most valuable goods to $i$. Summing over $i \in [n]$, we get 
\[
n \cdot \sum_{i \in [n]} v_i(B_i) = n \cdot \sw(\calB) \geq \sum_{i \in [n]} v_i([m]) - \sum_{i \in [n]} \sum_{g \in G_i} v_i(g).
\]
Substituting \Cref{ineq:max-matching}, we get the desired bound
\begin{align*}
n \cdot \sw(\calB) + n \sw(\calB) \ge \sum_{i \in [n]} v_i([m]) - \sum_{i \in [n]} \sum_{g \in G_i} v_i(g) + \sum_{i \in [n]} \sum_{g \in G_i} v_i(g) = \sum_{i \in [n]} v_i([m]),
\end{align*}
which implies $\sw(\calB) \geq \frac{1}{2n} \sum_{i \in [n]} v_i([m])$, as needed. 
\end{proof}

\subsection{The Case of High Optimal Welfare}\label{sec:ef1-high-opt}

As noted in \Cref{sec:ef1-absolute}, \Cref{lemma:efone-half} allows us to focus on fair division instances with scaled subadditive valuations in which the optimal social welfare is $\OPT = \Omega(\sqrt{n})$. This allows us to sacrifice $O(\sqrt{n})$ of the welfare in $\OPT$, obtain $O(\sqrt{n})$ approximation to the \emph{remaining} welfare through an $\EFone$ allocation, and yet achieve $O(\sqrt{n})$ approximation to $\OPT$. In particular, we present an algorithm that efficiently finds an $\EFone$ allocation $\calA$ with social welfare $\sw(\calA) \ge \frac{\OPT-2\sqrt{n}}{12\sqrt{n}}$. 

Ideally, our algorithm would like to use as reference an allocation $\calW^*$ with the optimal social welfare $\OPT$. However, for subadditive valuations, computing such an allocation is known to be NP-hard under both value queries and demand queries~\cite{DNS10}. Using a polynomial number of value queries, it is known that $\Theta(\sqrt{m})$ is the best possible approximation to the optimal social welfare~\cite{DNS10,mirrokni2008tight}. But this is too loose for our purpose. Hence, we instead turn to demand queries, for which an algorithm of \citet{feige2009maximizing} yields the optimal $2$-approximation. In particular, it efficiently computes an allocation $\calW$ with social welfare at least $\nicefrac{1}{2} \cdot \OPT$.\footnote{We note that this allocation only serves as a reference in our algorithm, and may not be $\EFone$ itself.} Outside of the black-box use of Feige's algorithm, the rest of our algorithm uses value queries (which are a special case of demand queries). We note that for the special case of submodular valuations, it is possible to efficiently compute an allocation with social welfare at least $\nicefrac{1}{2}\cdot \OPT$ using only value queries~\cite{lehmann2006combinatorial}; hence, for such special cases, our algorithm would not require access to demand queries. Even more importantly, we emphasize that our use of value or demand queries is only for the efficiency purpose; our main result --- the price of $\EFone$ --- is existential and independent of any query model. 

Starting with the high-welfare allocation $\calW$ returned by Feige's algorithm, our algorithm works as follows. It first indexes the $m$ goods as $g_1, g_2, \ldots, g_m$ such that the goods in each $W_i$ receive consecutive indices.\footnote{For example, we can index the goods such that $W_i = \set{g_k : 1 + \sum_{j=1}^{i-1} |W_{j}| \le k \le \sum_{j=1}^i |W_i|}$ for each agent $i \in [n]$.} Alternatively, consider a line graph $L = ([m],E)$ over the set of goods with edges $E = \set{(g_k,g_{k+1}) : k \in [m-1]}$. Then, each $W_i$ induces a connected subgraph of $L$. 

\begin{definition}
	Let $L = ([m],E)$ be a line graph over the goods. We say that $S \subseteq [m]$ is a \emph{connected bundle} in $L$ if $S$ induces a connected subgraph of $L$. Given a partial allocation $\calP$, define $\calU(\calP)$ as the set of connected components of $L$ that remain after removing the allocated goods $\cup_{i \in [n]} P_i$. We refer to $U \in \calU(\calP)$ as an \emph{unassigned connected bundle}. 
\end{definition}

\begin{figure}[h]
	\begin{center}
		\includegraphics[scale=.7]{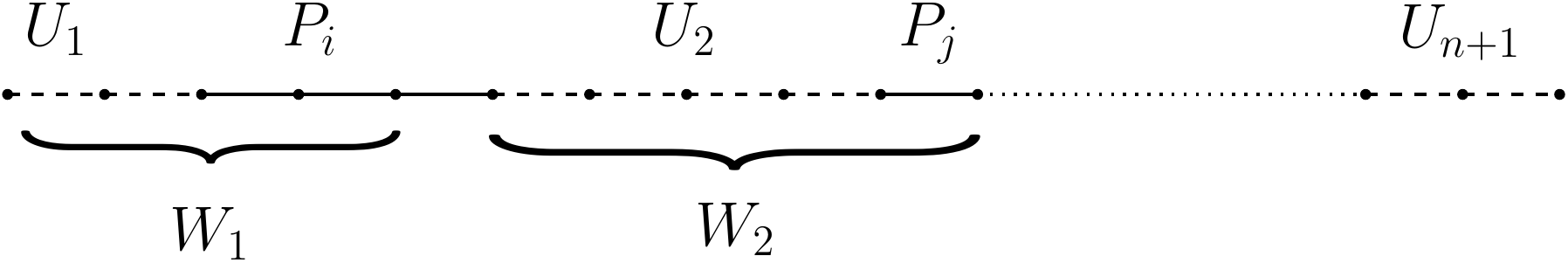}
	\end{center}
	\caption{Line graph over the goods and unassigned components.} 
	\label{figure:line-graph} 
\end{figure}

In this terminology, each $W_i$ is a connected bundle in $L$ (see \Cref{figure:line-graph}). \algefone builds a partial allocation $\calP$ by giving each agent $i$ her most valuable good from $W_i$. This allocation trivially satisfies two properties: it is $\EFone$, and each $P_i$ is a connected bundle in $L$. The algorithm then iteratively updates $\calP$ to improve its social welfare while maintaining both these properties. This iterative process is inspired by a similar algorithm for (divisible) cake-cutting by~\citet[Algorithm 1]{arunachaleswaran2019fair}.

In particular, at every iteration, our algorithm computes the set of unassigned connected bundles $\calU(\calP)$; note that removing $n$ connected bundles from $\calP$ can create at most $n+1$ unassigned connected bundles, as shown in \Cref{figure:line-graph}. If there is an unassigned connected bundle $U \in \calU(\calP)$ that some agent envies, then the algorithm finds an \emph{inclusion-wise minimal} subset of $U$ that is envied and allocates it to an envious agent. While this preserves exact envy-freeness in the cake-cutting setting of \citet{arunachaleswaran2019fair}, in our setting we argue that inclusion-wise minimality preserves the $\EFone$ property of $\calP$. We also argue that this iterative process terminates at a partial allocation $\calP$ that satisfies the desired social welfare guarantee. Finally, we use the algorithm of \citet{LMMS04} to extend this partial $\EFone$ allocation into a complete $\EFone$ allocation without losing social welfare (while the algorithm of Lipton et al. may not preserve the partial allocation to each agent, it does ensure that the utility of each agent from her final allocation is at least that from her partial allocation; see, e.g.,~\cite{AmanatidisMN20}). The detailed algorithm is presented as \Cref{alg:efone}. 

\begin{algorithm}[!htb]
  \caption{\algefone} \label{alg:efone}  
  \textbf{Input:} A fair division instance $\calI= \langle [n], [m], \set{v_i}_{i \in [n]} \rangle$ with demand-query oracle access to the subadditive valuations $v_i$s.
  
  \textbf{Output:} An $\EFone$ allocation $\calA$ with social welfare $\sw(\calA) \ge \frac{\OPT}{12\sqrt{n}} - \frac{1}{6}$.
  
  \begin{algorithmic}[1]
		\STATE Use Feige's algorithm~\cite{feige2009maximizing} to compute an allocation $\calW$ with social welfare $\sw(\calW) \ge \nicefrac{1}{2} \cdot \OPT$ \label{step:feige}
		
		\STATE Re-index the $m$ goods as $g_1, g_2, \ldots, g_m$ such that in the line graph $L$ over the goods containing edges $(g_k,g_{k+1})$ for $k \in [m-1]$, each $W_i$ forms a connected subgraph \label{step:define-L}
		
		\STATE Initialize $t = 0$
		
		\STATE For each agent $i \in [n]$ with $W_i \neq \emptyset$, pick $g^*_i \in \argmax{g \in W_i} v_i(g)$ and set $P^t_i = \set{g^*_i}$ \label{step:initialize}
		
		\STATE For each agent $i \in [n]$ with $W_i = \emptyset$, set $P^t_i = \emptyset$
		
		\COMMENT{For a partial allocation $\calP^t$, let $\calU(\calP^t)$ denote the collection of connected components in the line graph $L$ that remain after the removal of $\cup_{i \in [n]} P^t_i$}
		
		\WHILE{$\exists$ agent $i \in [n]$ and connected component $U \in \calU(\calP^t)$ with $v_i(P^t_i) < v_i(U)$}
		
			\STATE Let $U$ consist of goods $\set{g_a, g_{a+1}, \ldots, g_b}$
			
			\STATE Let $c \in [a,b]$ be the smallest index such that $v_k(P^t_k) < v_k(\set{g_a, g_{a+1}, \ldots, g_c})$ for some $k \in [n]$ \label{step:move-knife}
			
			\COMMENT{The choice of $c$ ensures that no agent $k$ values $\set{g_a,\ldots,g_{c-1}}$ more than her current bundle, i.e., $v_k(P^t_k) \ge v_k(\set{g_a,\ldots,g_{c-1}})$ for every agent $k$.}
			
			\STATE Pick an arbitrary agent $k$ for which $v_k(P^t_k) < v_k(\set{g_a, g_{a+1}, \ldots, g_c})$
			
			\STATE Set $P^{t+1}_k = \set{g_a, \ldots, g_c}$ and $P^{t+1}_j = P^t_j$ for all $j \neq k$ \label{step:update}
			
			\STATE Update $t \leftarrow t +1$
		
		\ENDWHILE
		
		\COMMENT{At this point, $v_i(P^t_i) \geq v_i(U)$ for all $i \in [n]$ and $U \in \calU(\calP^t)$. We will show that $\calP^t$ is a partial $\EFone$ allocation.}
		
		\STATE Use the algorithm of \citet{LMMS04} to extend the partial $\EFone$ allocation $\calP^t$ into a complete $\EFone$ allocation $\calA$ such that $v_i(A_i) \ge v_i(P^t_i)$ for each $i \in [n]$ \label{step:Lipton}
        
        \RETURN Allocation $\calA$
\end{algorithmic}
\end{algorithm}

Let us start by proving some relevant properties of \algefone. 

\begin{lemma}
\label{lemma:algefone-invariant}
When \algefone is run on a fair division instance $\calI= \langle [n], [m], \set{v_i}_{i \in [n]}  \rangle$ with subadditive valuations, the following hold regarding the partial allocation $\calP^t$ constructed after $t$ iterations of the while loop.
\begin{enumerate}
	\item If $t \ge 1$, then for each agent $i \in [n]$, either $P^t_i = P^{t-1}_i$ or $v_i(P^t_i) > v_i(P^{t-1}_i)$. 
	\item For each agent $i \in [n]$, $P_i$ is a connected bundle under the line graph $L$ constructed in \Cref{step:define-L}. 
	\item $\calP^t$ is $\EFone$. 
\end{enumerate}
\end{lemma}
\begin{proof}
	The first property trivially holds because \Cref{step:update} updates the bundle of exactly one agent $k$ in each iteration, and $k$ is chosen such that she values her new bundle strictly more than her previously assigned bundle. The second property also trivially holds because \Cref{step:update} always assigns a connected bundle to an agent. 
		
	For the third property, we use induction on $t$. For the base case of $t=0$, note that $\calP^0$ is trivially $\EFone$ because $|P^0_i|\le 1$ for each $i \in [n]$. Suppose $\calP^t$ is $\EFone$ for some $t \in \bbN$. To see why $\calP^{t+1}$ is $\EFone$, suppose agent $k$ is assigned $\set{g_a,\ldots,g_c}$ in iteration $t+1$. Because $v_k(\set{g_a,\ldots,g_c}) > v_k(P^t_k)$ by design, the only envy that may exist under $\calP^{t+1}$ is from an agent $i \neq k$ toward agent $k$. However, due to the minimality of $c$ chosen in \Cref{step:move-knife}, no agent $i$ can envy the bundle $\set{g_a,\ldots,g_{c-1}}$. Hence, agent $i$ may only envy agent $k$ by up to one good ($g_c$ in particular) under $\calP^{t+1}$. Hence, $\calP^{t+1}$ is $\EFone$. 
\end{proof}

\begin{lemma}\label{lemma:algefone-termination}
	When $\algefone$ is run on a fair division instance $\calI= \langle [n], [m], \set{v_i}_{i \in [n]}  \rangle$ with subadditive valuations, the following hold.
	\begin{enumerate}
		\item The while loop terminates after $T = O(n m^2)$ iterations.
		\item The partial allocation $\calP^T$ constructed at the end of the while loop satisfies $v_i(P^T_i) \ge v_i(U)$ for every agent $i \in [n]$ and unassigned connected bundle $U \in \calU(\calP^T)$. 
	\end{enumerate}
\end{lemma}
\begin{proof}
	Suppose \Cref{step:update}, in iteration $t$ of the while loop, updates the allocation to agent $k$ by assigning her the bundle $\set{g_a,\ldots,g_c}$. Then, by \Cref{lemma:algefone-invariant},  $\set{g_a,\ldots,g_c}$ must be a connected bundle in the line graph $L$ defined in \Cref{step:define-L} and agent $k$ must strictly prefer this bundle to her previously assigned bundle. 
	
	Because there are $O(m^2)$ possible connected bundles in $L$, and because an agent's value for her assigned bundle strictly increases every time her assignment is updated, an agent's assignment can only be updated $O(m^2)$ times. Because each iteration of the while loop updates the assignment of one agent, there can be at most $O(n m^2)$ iterations of the while loop. 
	
	The second statement follows directly from the condition of the while loop; the while loop would not have terminated if $v_i(P^T_i) < v_i(U)$ for any agent $i \in [n]$ and any unassigned connected bundle $U \in \calU(\calP^T)$. 	
\end{proof}

Using \Cref{lemma:algefone-invariant,lemma:algefone-termination}, we can derive the key technical result of this section. 

\begin{lemma}\label{lemma:apx-guarantee}
Given a fair division instance $\calI= \langle [n], [m], \set{v_i}_{i \in [n]} \rangle$ with scaled subadditive valuations, \algefone terminates in polynomial time and returns an $\EFone$ allocation $\calA$ satisfying 
\begin{align*}
\sw(\calA) \geq \frac{\OPT}{12\sqrt{n}} - \frac{1}{6}, 
\end{align*}
where $\OPT$ is the optimal social welfare achievable in instance $\calI$. 
\end{lemma}
\begin{proof}
Recall that allocation $\calW$ returned by Feige's algorithm in \Cref{step:feige} satisfies $\sw(\calW) \ge \frac{1}{2} \OPT$. Let the while loop in \algefone terminate with partial allocation $\calP$; that is, $\calP = \calP^T$, where $T$ is the number of iterations of the while loop. \Cref{lemma:algefone-termination} ensures that the while loop terminates, and therefore, $\calP$ is well-defined. We want to show that 
\begin{align}\label{ineq:partial-alloc}
\sqrt{n} + 6 \sqrt{n} \cdot \sw(\calP) \ge \sw(\calW) 
\end{align}

By \Cref{lemma:algefone-invariant}, $\calP$ is $\EFone$. Also, note that the algorithm of \citet{LMMS04} used in \Cref{step:Lipton} can extend a partial $\EFone$ allocation $\calP$ into a complete $\EFone$ allocation $\calA$ with $\sw(\calA) \ge \sw(\calP)$. Hence, \Cref{ineq:partial-alloc}, along with the fact that $\sw(\calW) \ge \frac{1}{2} \OPT$, completes the proof of the stated claim.  
 
Let us now prove that \Cref{ineq:partial-alloc} holds. First, due to \Cref{lemma:algefone-invariant}, each $P_i$ is a connected bundle in the line graph $L$. Hence, removing the goods allocated in $\calP$ creates at most $n+1$ connected components, i.e., $|\calU(\calP)| \le n+1$. Let $\calT(\calP) = \set{P_j : j \in [n]} \cup \calU(\calP)$ be the collection of assigned and unassigned connected bundles in $L$ created by $\calP$. Hence, $|\calT(\calP)| \le 2n+1$. For each agent $i \in [n]$, let $\calQ_i$ denote the bundles in $\calT(\calP)$ that intersect with $W_i$, i.e., $\calQ_i \coloneqq \set{ S \in \calT(\calP) : S \cap W_i \neq \emptyset}$. Write $q_i \coloneqq |\calQ_i|$. For instance, in \Cref{figure:line-graph}, $\calQ_1 = \set{U_1, P_i}$ and $q_1 = 2$.

Note that $\calT(\calP)$ forms a partition of $L$ into at most $2n+1$ connected components, while $\calW$ forms a partition of $L$ into at most $n$ connected components. This, along with the observation that $L$ is a line graph, can be used to bound the total number of intersections between the two. Specifically, we can move a marker from left to right, and keep track of the bundle from $\calW$ and the bundle from $\calT(\calP)$ that contain the current good (and thus intersect). Every time the marker encounters a \emph{new} intersection, it must have entered either a new bundle from $\calW$ or a new bundle from $\calT(\calP)$. Thus, we have
\begin{align}\label{ineq:intersection-count}
\textstyle\sum_{i \in [n]} q_i \leq (2n+1)+n = 3n+1.
\end{align}

Next, we partition the set of agents $[n]$ based on their $q_i$ value. Let $H \coloneqq \set{i \in [n] : q_i > 3 \sqrt{n}}$ and $H^c \coloneqq \set{i \in [n] : q_i \leq 3 \sqrt{n}}$. \Cref{ineq:intersection-count} immediately implies that $|H| \leq \sqrt{n}$. This, along with the observation that the valuations are scaled, gives us 
\begin{align}\label{ineq:H-welfare}
\textstyle\sum_{i \in H} v_i(W_i) \leq |H| \leq \sqrt{n}. 
\end{align}

Next, we bound $v_i(W_i)$ for $i \in H^c$. By the definition of $\calQ_i$, we have $v_i(W_i) = \sum_{S \in \calQ_i} v_i(S \cap W_i)$. We break the summation into two parts depending on whether $S \in \calU(\calP)$ or $S \in \calP$. 
\begin{enumerate}
	\item If $S \in \calU(\calP)$, note that from \Cref{lemma:algefone-termination}, we have $v_i(S \cap W_i) \le v_i(S) \le v_i(P_i) \le 2v_i(P_i)$. 
	\item If $S \in \calP$, then $S = P_j$ for some agent $j \in [n]$. Let us write $\hat{P}_j = P_j \cap W_i$. By the definition of $\calQ_i$, $\hat{P}_j \neq \emptyset$. Because $\calP$ is $\EFone$ (\Cref{lemma:algefone-invariant}) and $v_i$ is monotonic, agent $i$ does not envy the bundle $P_j$ (and thus the bundle $\hat{P}_j$) up to one good in it. Thus, there exists a good $\hat{g} \in \hat{P}_j$ such that $v_i(P_i) \geq v_i(\hat{P}_j \setminus \set{\hat{g}}) \geq v_i(\hat{P}_j) - v_i(\hat{g})$, where the last transition uses subadditivity of $v_i$. Further, $\hat{g} \in \hat{P}_j \subseteq W_i$. The initialization in \Cref{step:initialize} and \Cref{lemma:algefone-invariant} ensure that $v_i(\hat{g}) \le v_i(P^0_i) \le v_i(P_i)$. Combining with the previous step, we have $v_i(S \cap W_i) = v_i(\hat{P}_j) \le 2 v_i(P_i)$.
\end{enumerate}
Thus, for $i \in H^c$, we have $v_i(W_i) \le |\calQ_i| \cdot 2 v_i(P_i) \le 3\sqrt{n} \cdot 2 v_i(P_i)$. Thus, with \Cref{ineq:H-welfare}, we get
\[
\textstyle\sw(\calW) \le \sqrt{n} + 6\sqrt{n} \cdot \sum_{i \in H^c} v_i(P_i) \le \sqrt{n} + 6\sqrt{n} \cdot \sw(\calP),
\]
which is the desired \Cref{ineq:partial-alloc}, and the proof is complete.
\end{proof}

\subsection{Proof of Theorem~\ref{theorem:efone}}\label{section:proof-of-efone}
We are now ready to restate and prove \Cref{theorem:efone}. The upper bounds are established using \Cref{lemma:efone-half,lemma:apx-guarantee}. As mentioned earlier, the lower bound for scaled valuations is already obtained by \citet{BeiLMS19}. 

\TheoremEFone*
\begin{proof}
Let us first prove the upper bounds. For unscaled valuations, \Cref{lemma:efone-half} shows that there exists an $\EFone$ allocation $\calB$ with $\sw(\calB) \ge \frac{1}{2n} \sum_{i \in [n]} v_i([m]) \ge \frac{1}{2n} \OPT$, which yields the desired $O(n)$ bound. 

For an instance $\calI$ with scaled valuations, we consider two cases: either  $\OPT \leq 8 \sqrt{n}$ or $\OPT > 8 \sqrt{n}$. 

\paragraph{Case 1:} Suppose $\OPT \le 8 \sqrt{n}$. Because valuations are scaled (i.e. $v_i([m]) = 1$ for each agent $i \in [n]$), \Cref{lemma:efone-half} implies that the $\EFone$ allocation $\calB$ returned by \algefoneabs satisfies $\sw(\calB) \ge \frac{1}{2n} \sum_{i \in [n]} v_i([m])  = \frac{1}{2}$. Hence, $\sw(\calB) \ge \frac{1}{16\sqrt{n}} \OPT$.

\paragraph{Case 2:} Suppose $\OPT > 8 \sqrt{n}$. Then, $\frac{1}{6} < \frac{\OPT}{48\sqrt{n}}$. Now, \Cref{lemma:apx-guarantee} implies that the $\EFone$ allocation $\calA$ returned by \algefone satisfies 
\[
\sw(\calA) \ge \frac{\OPT}{12\sqrt{n}} - \frac{1}{6} \ge \frac{\OPT}{12\sqrt{n}} - \frac{\OPT}{48\sqrt{n}} = \frac{\OPT}{16\sqrt{n}}.
\]

Thus, in either case, there exists an $\EFone$ allocation with social welfare at least $\frac{1}{16\sqrt{n}} \OPT$, which yields the desired $O(\sqrt{n})$ bound. To obtain the required $\EFone$ allocation, we simply run both algorithms and take the allocation with higher social welfare.
	
We now prove the lower bounds. For scaled valuations, the desired $\Omega(\sqrt{n})$ bound is proved by \citet{BeiLMS19}. For unscaled valuations, consider an instance with $n$ agents and $m=n$ goods, and the following additive valuations: the first agent has value $n$ for each good, while the other agents have value $1/n$ for each good. The optimal social welfare is achieved by allocating all goods to the first agents, yielding $\OPT = n^2$. In contrast, that any $\EFone$ allocation gives exactly one good to each agent, thus achieving welfare less than $n+1$. This implies the desired $\Omega(n)$ bound.  
\end{proof}

Before concluding this section, we recall that for additive valuations, EF1 logically implies proportionality up to one good (Prop1). Hence, our upper bounds of $O(\sqrt{n})$ and $O(n)$ on the price of EF1 under scaled and unscaled additive valuations, respectively, carry over to the price of Prop1 as well. These bounds are tight: Our construction from the proof of \Cref{theorem:efone} for unscaled valuations and the construction due to \citet{BeiLMS19} for scaled valuations can be modified slightly for Prop1 as well. For completeness, we provide these modified constructions in the appendix. 

\begin{corollary}\label{cor:prop1}
	The price of Prop1 is $\Theta(\sqrt{n})$ for scaled additive valuations and $\Theta(n)$ for unscaled additive valuations.
\end{corollary}

%% file: mms.tex
\section{Price of $\nicefrac{1}{2}$-Approximate Maximin Share Guarantee}\label{sec:mms}

We now study the price of approximate $\MMS$. In this section, we limit our attention to additive valuations. Our main result, stated next, settles the price of \hmms for scaled and unscaled additive valuations.

\begin{restatable}{theorem}{TheoremMMS}
	\label{theorem:mms}
	The price of \hmms is $\Theta(\sqrt{n})$ for scaled additive valuations and $\Theta(n)$ for unscaled additive valuations.
\end{restatable}

We begin by proving the upper bound, which is based on \Cref{alg:mms-abs,alg:mms}. Recall that $\calW^* = (W^*_1, \ldots, W^*_n) \in \argmax{\calA \in \Pi_n([m])} \sw(\calA)$ denotes a social welfare maximizing allocation, and $\OPT = \sw(\calW^*)$ is the maximum social welfare. 

\subsection{An Absolute Welfare Guarantee}\label{sec:mms-absolute}

First, we show that when agents have additive valuations $\set{v_i}_{i \in [n]}$ (not necessarily scaled), there always exists a \hmms allocation $\calA$ with social welfare $\sw(\calA) \ge \frac{1}{3n} \sum_{i=1}^n v_i([m])$. In fact, \Cref{alg:mms-abs} computes one such allocation efficiently. 

The implication of this result is twofold. First, for both scaled and unscaled additive valuations, this establishes that the price of \hmms is $O(n)$ because, trivially, $\OPT \le \sum_{i=1}^n v_i([m])$. For unscaled valuations, this is precisely the bound we seek. For scaled valuations, we need to improve this to $O(\sqrt{n})$. To that end, note that for scaled valuations, $v_i([m])=1$ for each agent $i \in [n]$, so $\sw(\calA) \ge \nicefrac{1}{3} = \Omega(1)$. Thus, at least when $\OPT = O(\sqrt{n})$, this already establishes the desired upper bound of $O(\sqrt{n})$ on the price of \hmms, allowing us to limit our attention to instances with $\OPT = \Omega(\sqrt{n})$, as we do in \Cref{sec:mms-high-opt}.

\begin{algorithm}[htb!]
	\caption{\algmmsabs} \label{alg:mms-abs}  
	\textbf{Input:} A fair division instance $\calI= \langle [n], [m], \set{v_i}_i \rangle$ with additive valuations $\set{v_i}_{i \in [n]}$. \\
	\textbf{Output:} A \hmms allocation $\calB$ with $\sw(\calB) \ge \frac{1}{3n} \sum_{i=1}^n v_i([m])$. 
	\begin{algorithmic}[1]
		\setcounter{ALC@unique}{0}
		\small
		\STATE Initialize set of agents $A=[n]$, set of goods $G=[m]$, and bundles $B_i = \emptyset$, for all $i \in [n]$ 
		\WHILE{there exists agent $i \in A$ and good $g \in G$ such that $v_i(g) \geq \frac{1}{2|A|} v_i(G)$}
		\STATE Set $(i', g') \in \argmax{(i,g) \in A \times G \ : \ v_i(g) \geq \frac{1}{2|A|} v_i(G) } \ v_i(g) $ \label{step:arg-max}
		\STATE Set $B_{i'} =\{ g'\}$ and update $A \leftarrow A \setminus \{ i' \}$ along with $G \leftarrow G \setminus \{ g' \}$
		\ENDWHILE 
		
		\STATE Efficiently compute a Prop1 allocation $(B_i)_{i \in A}$ of the fair division instance $\tuple{A,G,\set{v_i}_{i \in A}}$ \label{step:prop1}
			
		\RETURN allocation $\calB = (B_1, B_2, \ldots, B_n)$
		
	\end{algorithmic}
\end{algorithm}

\Cref{alg:mms-abs} is a refinement of the algorithm of \citet{amanatidis2017approximation} for computing a \hmms allocation: in \Cref{step:arg-max}, we use an $\argmax{}$ to break ties, whereas they use an arbitrary pair $(i,g)$. Since \citet{amanatidis2017approximation} prove that their algorithm always returns a \hmms allocation, regardless of any tie-breaking, it follows that our refinement, \Cref{alg:mms-abs}, also always returns a \hmms allocation. It remains to prove that it provides the desired welfare guarantee as well. 

We note that there is another minor difference. In \Cref{step:prop1}, we compute an arbitrary Prop1 allocation of the remaining instance. While \citet{amanatidis2017approximation} use a specific method --- greedy round-robin --- that is guaranteed to return a Prop1 allocation, they note that their proof works for any Prop1 allocation. We state our algorithm with this generality in order to avoid introducing how greedy round-robin works. There are other methods that can be used to compute a Prop1 allocation~\cite{lipton2004approximately,CKMP+19,barman2019proximity}.

To establish the desired welfare bound, we first establish a lower bound on the agents' values for the goods that remain at various iterations of the algorithm. For the given instance $\langle [n], [m], \set{v_i}_{i \in [n]} \rangle$, we let $t$ denote the number of iterations of the while loop in \algmmsabs. We re-index the agents such that agents $1$ through $t$ receive, in order, a good (a singleton bundle) in this while loop. That is, each agent $i \in [t]$ is assigned a good in the $i^\text{th}$ iteration of the loop. The remaining $n-t$ agents receive a bundle through the Prop1 allocation computed in \Cref{step:prop1}. 

\begin{lemma}
	\label{lem:val-share}
	Consider a fair division instance $\calI = \tuple{ [n], [m], \set{v_i}_{i \in [n]} }$ with the re-indexing of the agents mentioned above, and $t$ denoting the number of agents assigned a good in the while loop of \algmmsabs. Let $\calB = (B_1, B_2, \ldots, B_n)$ denote the allocation returned by \algmmsabs. Then, we have 
	\begin{align*}
	\sum_{k=1}^{i-1} \frac{1}{n-k+1} \cdot v_k(B_k) \ + \ \frac{1}{n-i+1} \cdot v_i \left( [m] \setminus \cup_{k=1}^{i-1} B_k \right) & \geq \frac{1}{n} \cdot v_i([m]) \quad \text{for all $i \in [t]$, and} \\ 
	\sum_{k=1}^t \frac{1}{n-k+1} \cdot v_k(B_k) \ + \ \frac{1}{n-t} \cdot v_i \left( [m] \setminus \cup_{k=1}^t B_k \right) & \geq \frac{1}{n} \cdot v_i([m]) \quad \text{for all $i \in [n] \setminus [t]$.}
	\end{align*}
\end{lemma}
\begin{proof}
	For each agent $i \in [n]$, we establish the following inequality for every index $j \in \set{0,\ldots,\min\set{i-1, t}}$, via an induction on $j$. Then, instantiating $j=\min\set{i-1,t}$ yields the desired result.
	\begin{align}
	\sum_{k=1}^j \frac{1}{n-k+1} \cdot v_k(B_k) \ + \ \frac{1}{n-j} \cdot v_i \left( [m] \setminus \cup_{k=1}^{j} B_k \right) & \geq \frac{1}{n} \cdot v_i([m]). \label{ineq:induct}
	\end{align}
	
	Fix an arbitrary agent $i \in [n]$. The base case of $j=0$ is trivial, as it reduces to the inequality $\frac{1}{n} \cdot v_i([m]) \ge \frac{1}{n} \cdot v_i([m])$. 		For the induction hypothesis, suppose that \Cref{ineq:induct} holds with the index set to $j-1$. That is, we have 
	\begin{align}
	\sum_{k=1}^{j-1} \frac{1}{n-k+1} \cdot v_k(B_k) \ + \ \frac{1}{n-j+1} \cdot v_i \left( [m] \setminus \cup_{k=1}^{j-1} B_k \right) & \geq \frac{1}{n} \cdot v_i([m]). \label{ineq:ind-hyp}
	\end{align}
	
	To complete the proof via induction, we need to extend this inequality to index $j$. Note that $j \le \min\set{i-1,t} \le t$. Hence, agent $j$ is assigned a singleton bundle $B_j \coloneqq \set{g_j}$, and this assignment is made before $i$ receives a bundle. Consider two cases about the value that agent $i$ places on the good $g_j$.
	
	\medskip\noindent\emph{Case 1: $v_i(B_j) < \frac{1}{n-j+1} v_i \left( [m] \setminus \cup_{k=1}^{j-1} B_k \right)$.} Note that before agent $j$ was assigned $B_j$, there were $n-j+1$ agents left. In this case, agent $i$'s value for $B_j$ is below the ``average value'' of the remaining goods divided into $n-j+1$ bundles. Hence, the ``average'' goes up after assigning $B_j$ to agent $j$, i.e., 
	\begin{equation}\label{ineq:average-up}
	\frac{1}{n-j} v_i \left( [m] \setminus \cup_{k=1}^{j} B_k \right)  \geq \frac{1}{n-j+1} v_i \left( [m] \setminus \cup_{k=1}^{j-1} B_k \right).
	\end{equation}
	To see this formally, let us rearrange the inequality of this case as 
	\[
	(n-j+1) \cdot v_i(B_j) < v_i \left( [m] \setminus \cup_{k=1}^{j-1} B_k \right),
	\]
	or equivalently, 
	\[
	(n-j) \cdot v_i(B_j) < v_i \left( [m] \setminus \cup_{k=1}^{j} B_k \right).
	\]
	Adding $(n-j) \cdot v_i \left( [m] \setminus \cup_{k=1}^{j} B_k \right)$ on both sides and dividing both sides by $(n-j) \cdot (n-j+1)$ yields \Cref{ineq:average-up}. 
	
	Combining this with the induction hypothesis, we get that in \Cref{ineq:ind-hyp}, changing the index from $j-1$ to $j$ adds $\frac{1}{n-j+1} \cdot v_j(B_j) \ge 0$ to the first term in LHS, and due to \Cref{ineq:average-up}, the second term also weakly increases. Hence, the LHS weakly increases, and by the induction hypothesis, stays at least $\frac{1}{n} \cdot v_i([m])$. 
		
	\medskip\noindent\emph{Case 2: $v_i(B_j) \geq \frac{1}{n-j+1} v_i \left( [m] \setminus \cup_{k=1}^{j-1} B_k \right)$.} This condition implies that the pair $(i, g_j)$ was considered in \Cref{step:arg-max} of the algorithm in iteration $j$ of the while loop. The argmax selection criterion of this step implies that $v_j(B_j) \geq v_i(B_j)$. Now, we have
	\begin{align*}
	&\sum_{k=1}^{j} \frac{1}{n-k+1} \cdot v_k(B_k) \ + \ \frac{1}{n-j} \cdot v_i \left( [m] \setminus \cup_{k=1}^{j} B_k \right)\\
	&\ge \sum_{k=1}^{j} \frac{1}{n-k+1} \cdot v_k(B_k) \ + \ \frac{1}{n-j+1} \cdot v_i \left( [m] \setminus \cup_{k=1}^{j} B_k \right)\\
	&= \sum_{k=1}^{j-1} \frac{1}{n-k+1} \cdot v_k(B_k) \ + \frac{1}{n-j+1} \cdot v_j(B_j)\ +\ \frac{1}{n-j+1} \cdot v_i \left( [m] \setminus \cup_{k=1}^{j} B_k \right)\\
	&\ge \sum_{k=1}^{j-1} \frac{1}{n-k+1} \cdot v_k(B_k) \ + \frac{1}{n-j+1} \cdot v_i(B_j)\ +\ \frac{1}{n-j+1} \cdot v_i \left( [m] \setminus \cup_{k=1}^{j} B_k \right)\\
	&= \sum_{k=1}^{j-1} \frac{1}{n-k+1} \cdot v_k(B_k) \ +\ \frac{1}{n-j+1} \cdot v_i \left( [m] \setminus \cup_{k=1}^{j-1} B_k \right)\\
	&\ge \frac{1}{n} \cdot v_i([m]),
	\end{align*}
	where the third transition follows due to $v_j(B_j) \ge v_i(B_j)$, and the last inequality follows due to the induction hypothesis (\Cref{ineq:ind-hyp}).
	
	Setting $j = \min\set{i-1,t}$ for each agent $i$ in \Cref{ineq:induct} completes the proof.
\end{proof}

We now use \Cref{lem:val-share} to derive the desired guarantees for \algmmsabs (\Cref{alg:mms-abs}).

\begin{lemma}\label{lemma:mms-half}
	Given any fair division instance $\calI = \tuple{ [n], [m], \set{v_i}_{i \in [n]} }$ with additive valuations, \algmmsabs (\Cref{alg:mms-abs}) efficiently computes a \hmms allocation $\calB$ with social welfare $\sw(\calB) \ge \frac{1}{3n} \sum_{i=1}^n v_i([m])$. 
\end{lemma}

\begin{proof} 
	The fact that the algorithm runs efficiently and computes a \hmms allocation is due to \citet{amanatidis2017approximation}, as argued above. It remains to show that the stated welfare bound holds. 
	
	Recall that $t$ denotes the number of iteration of the while loop in \algmmsabs, and we re-indexed the agents so that each agent $i \in [t]$ is assigned a singleton bundle $B_i = \set{g_i}$ in iteration $i$ of the while loop. 
	
	For an agent $i \in [t]$, the selection in \Cref{step:arg-max} of the algorithm in $i^\text{th}$ iteration of the while loop ensures that $v_i(B_i) \geq \frac{1}{2(n-i+1)} \cdot v_i \left( [m] \setminus \cup_{k=1}^{i-1} B_k\right)$. Substituting this into \Cref{lem:val-share}, we have 
	\begin{align}\label{ineq:for-singleton}
	\sum_{k=1}^{i-1} \frac{1}{n-k+1} \cdot v_k(B_k) \ + \ 2 \cdot v_i(B_i) \geq \frac{1}{n} \cdot v_i([m]) \quad \text{for all $i \in [t]$}. 
	\end{align} 
	
	Similarly, fix an agent $i \in [n]\setminus[t]$. Note that at the termination of the while loop, we have $A = [n]\setminus[t]$, $G = [m] \setminus \cup_{k=1}^t B_k$ and $v_i(g) < \frac{1}{2(n-t)} v_i(G)$ for each $g \in G$. Because \Cref{step:prop1} computes a Prop1 allocation of $G$ to agents in $A$, we have that for some $g \in G$, $v_i(B_i) + v_i(g) \ge \frac{1}{n-t} v_i(G)$. The inequality $v_i(g) < \frac{1}{2(n-t)} v_i(G)$ yields $v_i(B_i) > \frac{1}{2(n-t)} v_i(G)$. Substituting this into \Cref{lem:val-share}, we have 
	\begin{align}\label{ineq:past-singleton}
	\sum_{k=1}^{t} \frac{1}{n-k+1} \cdot v_k(B_k) \ + \ 2 \cdot v_i(B_i) \geq \frac{1}{n} \cdot v_i([m]) \quad \text{for all $i \in [n] \setminus [t]$}. 
	\end{align}
	
	Summing \Cref{ineq:for-singleton,ineq:past-singleton} across all agents yields the desired welfare bound:
	\begin{align*}
	\frac{1}{n} \sum_{i=1}^n v_i([m]) &  \leq \sum_{i=1}^t \left( 2 \cdot v_i(B_i) + \sum_{k=1}^{i-1} \frac{1}{n-k+1} \cdot v_k(B_k) \right) + \sum_{i=t+1}^n \left( 2 \cdot v_i(B_i) + \sum_{k=1}^{t} \frac{1}{n-k+1} \cdot v_k(B_k) \right) \\ 
	& = 2 \cdot \sum_{i=1}^n v_i(B_i) + \sum_{i=1}^t \sum_{k=1}^{i-1} \frac{1}{n-k+1} \cdot v_k(B_k) + \sum_{i=t+1}^n \sum_{k=1}^{t} \frac{1}{n-k+1} \cdot v_k(B_k) \\
	& = 2 \cdot \sum_{i=1}^n v_i(B_i) + \sum_{k=1}^t \sum_{i=k+1}^n \frac{1}{n-k+1} \cdot v_k(B_k) \\
	& = 2 \cdot \sum_{i=1}^n v_i(B_i)  + \sum_{k=1}^t \frac{n-k}{n-k+1} \cdot v_k(B_k)  \\
	& \leq 2 \cdot \sum_{i=1}^n v_i(B_i)  + \sum_{k=1}^n v_k(B_k) = 3 \cdot \sum_{i=1}^n v_i(B_i).
	\end{align*}
	This completes the proof of the lemma. 
\end{proof}

\subsection{The Case of High Optimal Welfare}\label{sec:mms-high-opt}

As argued in \Cref{sec:mms-absolute}, \Cref{lemma:mms-half} allows us to restrict our attention to scaled additive valuations in which the optimal social welfare $\OPT = \Omega(\sqrt{n})$. Similarly to the case of EF1, we can now safely sacrifice $O(\sqrt{n})$ welfare, and simply achieve $O(\sqrt{n})$ approximation of the \emph{remaining} welfare. This is achieved by \algmms (\Cref{alg:mms}). 

However, \algmms requires knowledge of the maximin share $\MMS_i$ of each agent $i$. While computing this quantity is known to be strongly NP-hard, there exists a polynomial-time approximation scheme (PTAS) for it~\cite{woeginger1997polynomial}. For a fixed $\epsilon \in (0,1)$, this PTAS can compute an estimate $Z_i \in [(1-\epsilon) \MMS_i,\MMS_i]$ for each agent $i$ in polynomial time. We pass these estimates as input to \algmms, which runs in polynomial time and yields a $(\frac{1}{2}-\epsilon)$-$\MMS$ allocation with the desired welfare guarantee. 

We emphasize that the approximation introduced here is solely for computational purposes. To derive our main existential result about the price of $\hmms$, we can simply pass $Z_i = \MMS_i$ to $\algmms$, and it would return an exact $\hmms$ allocation with the desired welfare guarantee. 

\newcommand{\BIGCOMMENT}[1]{%
	\hspace{0.25in}$\Bigg\{$\ \begin{minipage}[c]{0.7\textwidth}#1\strut\end{minipage}$\Bigg\}$%
}

\newcommand{\LONGCOMMENT}[1]{%
	\STATE // {#1}%
}

\newcommand{\algorithmiccontinue}{\textbf{continue}}
\newcommand{\CONTINUE}{\STATE \algorithmiccontinue}

\begin{algorithm}[htb!]
	\caption{\algmms} \label{alg:mms}  
	\textbf{Input:} A fair division instance $\calI= \langle [m],[n], \{v_i\}_i \rangle$ with scaled additive valuations, and for a fixed $\epsilon \in [0,1)$, an estimate $Z_i \in [(1-\epsilon)\MMS_i,\MMS_i]$ for each agent $a$.
	
	\textbf{Output:} A $\left(\frac{1}{2}-\epsilon\right)$-$\MMS$ allocation $\mathcal{B}$ with $\sw(\mathcal{B}) \ge \frac{1}{3\sqrt{n}} \OPT - \frac{4}{3}$.
	
	\begin{algorithmic}[1]
		\small
		\STATE Compute a social welfare maximizing allocation $\calW^* \in \argmax{\calA \in \Pi_n([m])} \sw(\calA)$
		\STATE Index the goods as $g_1,\ldots,g_m$ so that, for each $i \in [n]$, the goods in $W^*_i$ receive consecutive indices \label{step:reindex}
		\STATE Initialize $\calB$ with $B_i = \emptyset$ for each $i \in [n]$
		\STATE Initialize $P = T = \emptyset$ 
		\STATE For each $i \in [n]$ with $\MMS_i = 0$, update $P \gets P \cup \set{i}$ if $v_i(W^*_i) = 0$, and $T \gets T \cup \set{i}$ otherwise. \COMMENT{$\MMS_i = 0$ iff agent $i$ has positive value for less than $n$ goods, which can be checked efficiently.} \label{step:zero-mms}
				\STATE Let $\Gsingle = \set{ i \in [n] \ : \ Z_i < \frac{2}{3 \sqrt{n}} v_i(W^*_i) \text{, and  there exists } g \in W^*_i \text{ such that } v_i(g) \geq \frac{1}{3 \sqrt{n}} v_i(W^*_i) }$
		\FOR{each $i \in \Gsingle$}
			\STATE Pick $g_i \in \argmax{g \in W^*_i} v_i(g)$, and set $B_i = \set{g_i}$ 
			\STATE Update $P \gets P \cup \set{i}$, and if $i \in T$, update $T \gets T \setminus \set{i}$ \label{step:gamma-single}
		\ENDFOR
		\WHILE{there exists an agent $a \in [n] \setminus (P \cup T)$ and a good $h \in [m] \setminus \cup_{b \in [n]} B_b$ such that $v_{a} (h) \geq \frac{1}{2} \cdot Z_a$} 
		\STATE Set $B_a = \set{h}$. 
		\STATE If $v_a(B_a) \geq \frac{1}{3 \sqrt{n}} v_a(W^*_a)$, update $P \gets P \cup \set{a}$, else update $T \gets T \cup \set{a}$ \label{step:first-while}
		\ENDWHILE 
		\vspace*{3pt} 
		\STATE Let $R \gets [m] \setminus \cup_{a \in [n]} B_a$ be the set of remaining goods
		\STATE Initialize $K = \emptyset$ and index $t = 1$ 
		\FOR{$t=1,\ldots,m$} \label{step:second-while}
		\IF{$g_t \notin R$}
			\CONTINUE
		\ENDIF
		\STATE Update $K \leftarrow K \cup \set{g_t}$
		\vspace*{3pt} 
		\IF{there exists an agent $i \in T$ such that $g_t \in W^*_i$ and $v_i(K) \geq \frac{1}{3 \sqrt{n}} v_i(W^*_i)$}  \label{step:if-T-to-P}
		\STATE Set $(B_i,K) \gets (K,B_i)$ \COMMENT{Swap $B_i$ and $K$} \label{step:swap}
		\STATE Update $P \leftarrow P \cup \set{i}$ and $T \leftarrow T \setminus \set{i}$ \label{step:T-to-P}
		\ENDIF
		\IF{there exists an agent $a \in [n] \setminus (P \cup T)$ such that $v_a(K) \geq \frac{1}{2} \cdot Z_a$}
		\STATE Set $B_a = K$ and update $K = \emptyset$ \label{step:assign}
		\STATE If $v_a(B_a) \geq \frac{1}{3 \sqrt{n}} v_a(W^*_a)$, update $P \gets P \cup \set{a}$, else update $T \gets T \cup \set{a}$ \label{step:second-while-assign}
		\ENDIF
		\ENDFOR
		\STATE Let $X = [m] \setminus \cup_{a \in [n]} B_a$ be the set of unassigned goods. Assign each $g \in X$ to agent $i$ such that $g \in W^*_i$, i.e., $B_i \leftarrow B_i \cup (W^*_i \cap X)$.  
		\RETURN allocation $(B_1, B_2, \ldots, B_n)$
	\end{algorithmic}
\end{algorithm}

\begin{lemma}\label{lem:alg-mms}
Given a fair division instance $\calI = \tuple{[n], [m], \set{v_i}_{i \in [n]}}$ with scaled additive valuations, and, for a fixed $\epsilon \in [0,1)$, an estimate $Z_i \in [(1-\epsilon)\MMS_i,\MMS_i]$ for each agent $i$, \algmms (\Cref{alg:mms}) efficiently computes a $(\frac{1}{2}-\epsilon)$-$\MMS$ allocation $\calB$ with the property that $3 \sqrt{n} \cdot \sw(\calB) + 4 \sqrt{n} \ge \OPT$; here, $\OPT$ denotes the optimal social welfare in $\calI$.
\end{lemma}

The intuition behind the algorithm is as follows. Our goal is to assign each agent $i$ a bundle of value at least $\frac{1}{2} \cdot Z_i \ge \left(\frac{1}{2}-\epsilon\right) \cdot \MMS_i$, and ensure that most (all but $\sqrt{n}$) agents $i$ achieve a value at least $\Omega\left(\frac{1}{\sqrt{n}} v_i(W^*_i)\right)$; here, $\calW=(W^*_1, \ldots, W^*_n)$ denotes a social welfare maximizing allocation. Therefore, throughout the algorithm, we keep track of two sets of agents, $T$ (temporary) and $P$ (permanent). Each agent $i \in T \cup P$ must have received a bundle $B_i$ worth $v_i(B_i) \ge \frac{1}{2} \cdot Z_i$, and agent $i \in P$ further has $v_i(B_i) \ge \frac{1}{3\sqrt{n}} \cdot v_i(W^*_i)$. The algorithm ensures that an agent is never removed after being added to $P \cup T$. Specifically, once she is added to $T$, she can only be moved to $P$. Additionally, once an agent is included in the set $P$, her assignment is never updated and she remains in $P$. 

Line \ref{step:zero-mms} and \ref{step:gamma-single} handle easy cases, when an agent $i$ either has $\MMS_i = 0$, or can be added directly to $P$ by giving her a single good from $W^*_i$. Line \ref{step:first-while} addresses agents to whom giving a single good (not necessarily from $W^*_i$) is sufficient to add them to $P \cup T$. These steps leverage the fact (proved below) that the maximin share is maintained while assigning away singleton bundles. 

Finally, Line \ref{step:second-while} leverages an idea similar to what we utilized for \EFone. This step slowly grows a bundle by iteratively adding goods ordered according to $\calW^*$, and assigns the bundle as soon as its value for some agent is at least half of her maximin share. Line \ref{step:T-to-P} plays a key role in bookkeeping, as we show in our proofs.

We partition the set of agents into three sets---$\Gmms$, $\Gsingle$, and $\Ghard$---where agent $i$ is placed into a set depending on how the estimate $Z_i$ of the maximin share $\MMS_i$ relates to $v_i(W^*_i)$.  
\begin{align*}
\Gmms &\coloneqq \set{i \in [n] \ : \ Z_i \geq \frac{2}{3 \sqrt{n}} v_i(W^*_i)},\\
\Gsingle &\coloneqq \set{ i \in [n] \ : \ Z_i < \frac{2}{3 \sqrt{n}} v_i(W^*_i) \text{, and  there exists } g \in W^*_i \text{ such that } v_i(g) \geq \frac{1}{3 \sqrt{n}} v_i(W^*_i) },\\
\Ghard &\coloneqq \set{ i \in [n] \ : \ Z_i < \frac{2}{3 \sqrt{n}} v_i(W^*_i) \text{, and  for all } g \in W^*_i \text{ we have } v_i(g) < \frac{1}{3 \sqrt{n}} v_i(W^*_i) }.
\end{align*}

For an agent $i \in \Gsingle$, giving her a single good from $\argmax{g \in W^*_i} v_i(g)$ gives her value at least $\frac{1}{3 \sqrt{n}} v_i(W^*_i) > \frac{1}{2} \cdot Z_i$. Thus, \algmms begins by giving each such agent $i$ her most valuable good from $W^*_i$ and adding her directly to $P$. Next, for an agent $i \in \Gmms$, giving her value at least $\frac{1}{2} \cdot Z_i$ also guarantees that her value is at least $\frac{1}{3 \sqrt{n}} v_i(W^*_i)$; thus, \algmms solely aims to give such agents $\frac{1}{2} \cdot Z_i$ value. The situation is more tricky for agents in $\Ghard$. 

Before we begin to analyze \algmms, we need the following lemma. This builds upon a similar result by \citet{amanatidis2017approximation}, who show that giving a subset of agents a single good each and removing these agents and their assigned goods from the instance cannot reduce the maximin share of any remaining agent. In our case, we need to consider assigning bundles that either consist of a single good or are of limited worth to an agent $\ell$, and we are interested in the value that agent $\ell$ must have for the remaining goods. This lemma will be useful later in analyzing \algmms. Let us first introduce some useful notation. For an agent $i \in [n]$, an integer parameter $k \in \bbN$, and a subset of goods $G \subseteq [m]$, define 
\begin{align*}
\textstyle\MMS_i^k (G) \coloneqq \max_{(P_1, \ldots, P_k) \in \Pi_k(G)} \ \min_{j \in [k]} \ v_i(P_j). 
\end{align*}
This is the highest value that agent $i$ can guarantee by splitting $G$ into $k$ bundles and receiving the worst bundle. Note that $\MMS_i = \MMS_i^n([m])$.

\begin{lemma}
\label{lemma:mms-flow}
Let $\calI = \tuple{[n], [m], \set{v_i}_{i \in [n]}}$ be a fair division instance with additive valuations. Let $\ell \in [n]$ be an agent. Let $\set{B_a}_{a \in S}$ be a collection of bundles assigned to a subset of agents $S \subseteq [n] \setminus \set{\ell}$ with the property that, for all $a \in S$, we either have $|B_a| = 1$ or $v_{\ell}(B_a) \leq \MMS_{\ell}$. Then, the value of agent $\ell$ for the unassigned goods satisfies 
\[
v_\ell \left( [m] \setminus \cup_{a \in S} B_a \right) \geq  \left( n-|S| \right) \cdot \MMS_\ell.
\]
\end{lemma}
\begin{proof}
\citet{amanatidis2017approximation} show that removing a single good along with a single agent from the instance does not decrease the maximin share of any remaining agent; in particular, for any integer $k \in \bbN$ and subset of goods $G \subseteq [m]$, the following inequality holds for all $g \in G$ and agents $\ell \in [n]$:
\begin{align}
\MMS^k_\ell(G) \leq \MMS_\ell^{k-1} \left( G \setminus \{ g \} \right)  \label{ineq:shift-mms}
\end{align} 

Among the set of agents $S$ considered in the lemma statement, write $S' \subseteq S$ to denote the ones that receive exactly one good, i.e., $S' \coloneqq \set{ a \in S \ : \ |B_a| = 1 }$. A repeated application of \Cref{ineq:shift-mms} gives us 
\[
\MMS_\ell^{n-|S'|} \left( [m] \setminus \cup_{a \in S'} B_a \right) \geq \MMS_\ell^n([m]) = \MMS_{\ell}.
\]

Furthermore, using the fact that the valuation $v_\ell$ is additive, we get 
\begin{align}
v_\ell\left( [m] \setminus \cup_{a \in S'} B_a \right) \geq (n- |S'|) \cdot \MMS_\ell^{n-|S'|} \left( [m] \setminus \cup_{a \in S'} B_a \right) \geq (n - |S'|) \cdot \MMS_\ell. \label{ineq:mms-step}
\end{align}

For each remaining agent $a \in S \setminus S'$, the assumption in the lemma's statement ensures that $v_\ell(B_a) \leq \MMS_\ell$. Therefore, 
\begin{align*}
v_\ell([m] \setminus \cup_{a \in S} B_a) & = v_\ell \left( [m] \setminus \cup_{a \in S'} B_a \right) - \sum_{a \in S \setminus S'} v_\ell(B_a) \tag{since $v_\ell$ is additive} \\
& \geq (n - |S'|) \cdot  \MMS_\ell - \sum_{a \in S \setminus S'} v_\ell(B_a) \tag{via \Cref{ineq:mms-step}} \\ 
& \geq  (n - |S'|) \cdot \MMS_\ell  - |S \setminus S'| \ \MMS_\ell \\
& = (n-|S|) \cdot \MMS_\ell,
\end{align*}
which is the desired inequality.
\end{proof} 

We now begin our analysis of \algmms. First, we show that certain invariants hold throughout the execution of the algorithm. 
\begin{lemma}\label{lem:algmms-invariants}
The following hold at any stage during the execution of \algmms
\begin{enumerate}
	\item $P \cap T = \emptyset$.
	\item If $i \in P$, $v_i(B_i) \ge \frac{1}{3\sqrt{n}} v_i(W^*_i)$ and $B_i$ is never updated afterwards.
	\item If $i \in P \cup T$, $v_i(B_i) \ge \frac{1}{2} \cdot Z_i$, and $i$ is never removed from $P \cup T$ afterwards. 
	\item $T \subseteq \Ghard$
\end{enumerate}
\end{lemma}
\begin{proof}
	The first two invariants are easy to verify by considering at the pseudocode of the algorithm. For the third invariant, note that once an agent $i$ is added to $P \cup T$, it is never removed (but can move from $T$ to $P$). 
	
	Next we show that when $i \in P \cup T$, we must have $v_i(B_i) \ge \frac{1}{2} \cdot Z_i$. Let us begin by showing that when $i$ is \emph{first} added to $P \cup T$, we have $v_i(B_i) \ge \frac{1}{2} \cdot Z_i$. This can happen in one of Lines \ref{step:gamma-single}, \ref{step:first-while}, or \ref{step:second-while-assign}. In each of these cases, one can verify that the condition for selection of $i$ ensures that $v_i(B_i) \ge \frac{1}{2} \cdot Z_i$ at the time it is included in $P \cup T$. We also need to consider Line \ref{step:T-to-P}, in which the bundle of an agent $i$, already added to $T$, can be updated. Note that, for such an agent $i$, we cannot have $i \in \Gsingle$, as agents in $\Gsingle$ are immediately added to $P$ in \Cref{step:gamma-single}. Similarly, we also cannot have $i \in \Gmms$ because in that case, $\frac{1}{2} \cdot Z_i \ge \frac{1}{3\sqrt{n}} \cdot v_i(W^*_i)$; hence, when $i$ was first added to $P \cup T$, it would have been added to $P$ directly. Therefore, for an agent $i$ considered in Line \ref{step:T-to-P} we must have $i \in \Ghard$. However, in that case, the selection criterion in Line \ref{step:T-to-P} implies that $v_i(B_i) \ge \frac{1}{3\sqrt{n}} \cdot v_i(W^*_i) \ge \frac{1}{2} \cdot Z_i$, as desired. 
	
	Finally, we show that $T \subseteq \Ghard$ throughout the execution of the algorithm. This is because each agent $i \in \Gsingle$ is added to $P$ at Line \ref{step:gamma-single}, and since $T \cap P = \emptyset$, we have that any agent $i \in \Gsingle$ cannot be in $T$. Similarly for agents in $\Gmms$, recall  that whenever an agent $i \in \Gmms$ is added to $P \cup T$ she must have received a bundle worth at least $\frac{1}{2}\cdot Z_i$ to her. Since $\frac{1}{2}\cdot Z_i \ge \frac{1}{3\sqrt{n}}\cdot v_i(W^*_i)$ for all $i \in \Gmms$, agent $i$ would be added included directly into $P$ at this point. Therefore, any agent that is placed in $T$ upon receiving a bundle must belong to the set $\Ghard$, i.e., the containment $T \subseteq \Ghard$ holds throughout the execution of \algmms. 	 
\end{proof}

Next, we show that the algorithm terminates in polynomial time, and certain desired properties hold at its termination. As we will see later, these properties are key to establishing the desired welfare guarantee of the algorithm. 
\begin{lemma}\label{lem:termination}
	\algmms terminates in polynomial time, and at its termination, the following hold.
	\begin{enumerate}
		\item $P \cup T = [n]$.
		\item $|T| \le 4\sqrt{n}$.
	\end{enumerate}
\end{lemma}
\begin{proof}
	\noindent\textbf{Termination:} To argue that the algorithm terminates, we must show that the while loop (\Cref{step:first-while}) terminates. This is easy to see because in each iteration of the loop, a new agent is added to $P \cup T$, and by the third invariant in \Cref{lem:algmms-invariants}, it is never removed from $P \cup T$ once added. It is easy to verify that the algorithm runs in polynomial time. 
	
	Let $\overline{P}$ and $\overline{T}$ denote the values of sets $P$ and $T$ at the termination of the algorithm, respectively. We want to show that $\overline{P} \cup \overline{T} = [n]$ and $|\overline{T}| \le 4\sqrt{n}$. 
	
	\medskip\noindent\textbf{$\mathbf{\overline{P} \cup \overline{T} = [n]}$:} Because $P \cup T$ grows monotonically (\Cref{lem:algmms-invariants}), we only need to show that each agent is included in $P \cup T$ at some point. Suppose, towards a contradiction, that agent $i$ is never added to $P \cup T$. 
	Since $i$ is not added to $P \cup T$ during the first while loop in Line \ref{step:first-while}, it must be the case that 
	\begin{equation}\label{eqn:small-single-good}
	v_i(g) < \frac{1}{2} \cdot Z_i, \forall g \in R.
	\end{equation}
	Otherwise, the while loop would not have terminated with goods in $R$ left unassigned. 
	
	We want to show that whenever an agent $j \neq i$ is assigned a nonempty bundle $B_j$ during the execution of the algorithm, $B_j$ satisfies the condition of \Cref{lemma:mms-flow}, i.e., either $|B_j| = 1$ or $v_i(B_j) \le \MMS_i$. We prove this by induction on the progress of the algorithm. Note that bundles assigned during Lines \ref{step:gamma-single} and \ref{step:first-while} are singleton, and therefore trivially satisfy the former condition. 
	
	Consider a bundle $K$ is assigned to an agent in the $t^{\text{th}}$ iteration of the for loop in Line \ref{step:second-while}, where $g_t \in R$ (that is, $g_t$ was unassigned at the start of the loop, and was thus added to $K$). If $g_t \notin K$, then the swap in Line \ref{step:second-while-assign} must have been executed. In this case, $K$ is the bundle that some agent was previously assigned, and the induction hypothesis implies that $K$ satisfies the condition of \Cref{lemma:mms-flow}. If $g_t \in K$, then we consider two cases. If $K = \set{g_t}$, then the former condition trivially holds. If $|K| \ge 2$, then note that $K \setminus \set{g_t}$ was not assigned in the previous iteration of the for loop. Since we had $i \in [n] \setminus (P \cup T)$, it must have been the case that $v_i(K \setminus \set{g_t}) < \frac{1}{2} \cdot Z_i$. Combining this with \Cref{eqn:small-single-good} and using the additivity of $v_i$, we have $v_i(K) \le Z_i \le \MMS_i$. Hence, the latter condition is satisfied. 
	
	Consider the set of unassigned goods $X$ at the end of the for loop in \Cref{step:second-while}. We have established that at this stage, for each $j \in P \cup T$, either $|B_j| = 1$ or $v_i(B_j) \le \MMS_i$. Using \Cref{lemma:mms-flow}, we have $v_i(X) \ge (n-|P \cup T|) \cdot \MMS_i \ge \MMS_i > 0$, where the penultimate inequality follows because $|P \cup T| \le n-1$ as $i \notin P \cup T$, and the final inequality follows because all agents $i$ with $\MMS_i = 0$ are added to $P \cup T$ in Line \ref{step:zero-mms} and then never removed from it (\Cref{lem:algmms-invariants}). Note that $v_i(X) > 0$ implies $X \neq \emptyset$. Then, in the last iteration of the for loop, we must have had $K = X$. This equality follows from the observation that as soon as any good gets unassigned it is included in the set $K$ (Line \ref{step:swap}). Hence, for every iteration count $t$, the set $K$ contains all unassigned goods from $\set{g_1,\ldots,g_t}$. Furthermore, note that at termination $K=X$ was not assigned to any remaining agent. 	
However, since $i \in [n] \setminus (P \cup T)$ and $v_i(X) \ge \MMS_i$, agent $i$ should have been assigned $K=X$ in Line \ref{step:second-while-assign}, which is the desired contradiction.

\medskip\noindent\textbf{$\mathbf{|\overline{T}| \le 4 \sqrt{n}}$:} \Cref{lem:algmms-invariants} ensures that the containment $T \subseteq \Ghard$ holds throughout the execution of the algorithm. Hence, in particular, we have $\overline{T} \subseteq \Ghard$.

Using this property and a charging argument, we will establish the inequality $|\overline{T}| \le \sqrt{n}$. We know from $\overline{P} \cup \overline{T} = [n]$ and monotonic growth of $P \cup T$ that each agent is added to $P \cup T$ exactly once during the execution of the algorithm. Define a partition $(C_1,\ldots,C_n)$ of the set of agents $[n]$ as follows. For each agent $j$, consider the bundle $B_j$ that she is assigned the first time she is added to $P \cup T$ (i.e., in one of Lines \ref{step:gamma-single}, \ref{step:first-while}, or \ref{step:second-while-assign}, but not in Line \ref{step:T-to-P}), and add $j$ to $C_i$, where agent $i$ is such that $W^*_i$ contains the highest indexed good in $B_j$, according to the indexing in Line \ref{step:reindex}. That is, the set $C_i$ populates agents $j$ whose initial bundle $B_j$ ends in $W^*_i$. Our goal is to show that for each $i \in \overline{T}$, $|C_i| \ge \sqrt{n}-1$, which (for $n \geq 2$) would imply $|\overline{T}| \le 4 \sqrt{n}$ .
	
	To see this, fix $i \in \overline{T}$, and consider a partition of $C_i$ into $C^1_i$ and $C^2_i$, where $C^1_i$ contains agents in $C_i$ who were first added to $P \cup T$ in either Line \ref{step:gamma-single} or Line \ref{step:first-while}, and $C^2_i$ contains agents in $C_i$ who were first added to $P \cup T$ in Line \ref{step:second-while-assign}. If $|C^1_i| \ge \sqrt{n}$, then trivially $|C_i| \ge \sqrt{n}$, and we are done. Thus, assume $|C^1_i| < \sqrt{n}$.
	
	Note that $C^1_i$ consists of all agents who were assigned a good from $W^*_i$ before the start of the for loop in Line \ref{step:second-while}. Since $i \in \overline{T} \subseteq \Ghard$, $v_i(g) < \frac{1}{3\sqrt{n}} \cdot v_i(W^*_i)$ for each $g \in W^*_i$. Since there were less than $\sqrt{n}$ such assignments, we have that for the goods $R$ left unassigned before the for loop,
	\[
	v_i(R \cap W^*_i) \ge v_i(W^*_i) - \sqrt{n} \cdot \frac{1}{3\sqrt{n}} \cdot v_i(W^*_i) \ge \frac{2}{3} \cdot v_i(W^*_i).
	\]
	
	Let $W^*_i = \set{g_t : t_1 \le t \le t_2}$ (the indexing in Line \ref{step:reindex} ensures that $W^*_i$ has consecutively indexed goods). We show that during iterations $t_1$ through $t_2$ of the for loop, at least $\sqrt{n}$ agents will be added to $C_i$. Note that because agent $i$ is never added to $P$, the if condition in Line \ref{step:if-T-to-P} must never be executed during these iterations. Hence, any bundle $K$ assigned in Line \ref{step:second-while-assign} during iterations $t_1$ through $t_2$ must be worth at most $\frac{2}{3\sqrt{n}} \cdot v_i(W^*_i)$ to agent $i$: if agent $i$ is in the set $T$, then the non-execution of the if-clause in Line \ref{step:if-T-to-P} gives us $v_i(K) < \frac{1}{3\sqrt{n}} \cdot v_i(W^*_i)$. Otherwise, even if $i \in \Ghard$ is not in $T$, then the removal of the highest-index good in $K$ must reduce its value below $\frac{1}{2} \cdot Z_i \leq \frac{1}{3 \sqrt{n}} v_i(W*_i)$. The definition of $\Ghard$ ensures this highest-index good is of value at most $\frac{1}{3 \sqrt{n}} v_i(W*_i)$, and hence, $v_i(K) < \frac{2}{3\sqrt{n}} \cdot v_i(W^*_i)$. 
	
	Since $v_i(R \cap W^*_i) \ge \frac{2}{3} \cdot v_i(W^*_i)$, at least $\sqrt{n}-1$ such assignments must be made during iterations $t_1$ through $t_2$, where we subtract $1$ for the last batch of goods from $W^*_i \cap R$ (again, worth at most $\frac{2}{3\sqrt{n}} \cdot v_i(W^*_i)$) that may be assigned after iteration $t_2$. Hence, we have $|C_i| \ge \sqrt{n}-1$ in this case as well. This completes the charging argument, and thus the proof that $\overline{T} \le 4 \sqrt{n}$. 
\end{proof}

We are now ready to establish the desired guarantees of \algmms from \Cref{lem:alg-mms}.

\begin{proof}[Proof of \Cref{lem:alg-mms}]
	The fact that \algmms terminates in polynomial time is proved in \Cref{lem:termination}. We need to show that it returns a $\left(\frac{1}{2}-\epsilon\right)$-$\MMS$ allocation with the desired welfare guarantee. 
	
	As before, let $\overline{P}$ and $\overline{T}$ denote the values of $P$ and $T$ at the termination of the algorithm, respectively. \Cref{lem:termination} gives us $\overline{P} \cup \overline{T} = [n]$, and \Cref{lem:algmms-invariants} ensures that for each $i \in \overline{P} \cup \overline{T} = [n]$, $v_i(B_i) \ge \frac{1}{2} \cdot Z_i \ge \frac{1}{2} \cdot (1-\epsilon) \cdot \MMS_i \ge \left(\frac{1}{2}-\epsilon\right) \cdot \MMS_i$. Thus, the allocation returned by \algmms is $\left(\frac{1}{2}-\epsilon\right)$-$\MMS$. 
	
	For the welfare guarantee, note that for each $i \in \overline{P}$, \Cref{lem:algmms-invariants} shows that $v_i(B_i) \ge \frac{1}{3 \sqrt{n}} \cdot v_i(W^*_i)$. And \Cref{lem:termination} shows that $|\overline{T}| \le 4 \sqrt{n}$. Combining them, we have 
	\begin{align*}
	\OPT = \sum_{i \in [n]} v_i(W^*_i) = \sum_{i \in \overline{P}} v_i(W^*_i) + \sum_{i \in \overline{T}} v_i(W^*_i) \leq 3 \sqrt{n} \sum_{i \in \overline{P}} v_i(B_i) + |\overline{T}| \leq 3 \sqrt{n} \sum_{i \in [n]} v_i(B_i) + 4 \sqrt{n},
	\end{align*}
	where the third transition holds because the valuations are scaled.
\end{proof}

\subsection{Proof of Theorem \ref{theorem:mms}}

Using the properties of \algmmsabs and \algmms from \Cref{sec:mms-absolute,sec:mms-high-opt}, we can directly prove the upper bounds on the price of \hmms in \Cref{theorem:mms}. For the lower bounds, we use constructions that appeared in the work of \citet{caragiannis2012efficiency}. We restate and prove the theorem below.

\TheoremMMS*
\begin{proof}
\noindent\textbf{Unscaled valuations:} \Cref{lemma:mms-half} shows that \algmmsabs returns a \hmms allocation $\calB$ with social welfare $\sw(\calB) \ge \frac{1}{3n} \sum_{i=1}^n v_i([m])$. Because the optimal social welfare is trivially upper bounded as $\OPT \le \sum_{i=1}^n v_i([m])$, we have that the price of \hmms is at most $3n =O(n)$. 

To show that this is tight, consider a fair division instance with $n$ agents with the following (unscaled) additive valuations over $n$ goods. Fix $\varepsilon \in (0,1)$. Let $v_1(j) = 1$ for each $j \in [n]$. For each $i \in [n] \setminus \set{1}$, $v_i(j) = \varepsilon$ for each $j \in [m]$. Thus, the goods are identical for all agents. It is easy to see that $\MMS_1 = 1$ while $\MMS_i = \varepsilon$ for each $i \in [n] \setminus \set{1}$. Hence, the only \hmms allocation is to give each agent a single good. This generates social welfare $1+(n-1)\varepsilon$, whereas the optimal social welfare is $n$, obtained by assigning all the goods to agent $1$. The ratio $n/(1+(n-1)\varepsilon)$ converges to $n$ as $\varepsilon \to 0$. Hence, the price of \hmms is at least $n = \Omega(n)$. 

\medskip\noindent\textbf{Scaled valuations:} Again, from \Cref{lemma:mms-half}, we know that \algmmsabs (with the input $Z_i = \MMS_i$ for each agent $i$) returns a \hmms allocation $\calB$ with $\sw(\calB) \ge \frac{1}{3n} \sum_{i=1}^n v_i([m]) = \frac{1}{3}$ for scaled valuations. Hence, if $\OPT \le 5\sqrt{n}$, then the \hmms allocation returned by \algmmsabs gives a $15\sqrt{n}$ approximation to the optimal social welfare.

Otherwise, suppose $\OPT \ge 5\sqrt{n}$. Then by \Cref{lem:alg-mms}, we know that \algmms returns an allocation with social welfare 
$
\sw(\calB) \ge \frac{\OPT-4\sqrt{n}}{3\sqrt{n}} \ge \frac{\OPT - \nicefrac{4}{5} \OPT}{3\sqrt{n}} = \frac{\OPT}{15\sqrt{n}}.
$
Hence, in this case, the \hmms allocation returned by \algmms provides a $15\sqrt{n}$ approximation to $\OPT$. 

We have established that in either case, the price of \hmms is at most $15\sqrt{n} = O(\sqrt{n})$. For the lower bound, we use a construction from \citet{caragiannis2012efficiency}, which later appeared in several other papers on the price of fairness. The instance has $n$ agents and goods, with the following scaled additive valuations:

\begin{itemize}
\item For $i \le \floor{\sqrt{n}}$, agent $i$ has value $1/\floor{\sqrt{n}}$ for the $\floor{\sqrt{n}}$ goods $[(i-1) \floor{\sqrt{n}} + 1$ to $i \floor{\sqrt{n}}]$, and value $0$ for the remaining goods. Call such an agent a \emph{high} agent.
\item For $i > \floor{\sqrt{n}}$, agent $i$ has value $1/n$ for all $n$ goods. Call such an agent a \emph{low} agent. 
\end{itemize}

It is easy to check that for each low agent $i$, $\MMS_i = 1/n$. Thus, any \hmms allocation must give each such agent at least one good. The social welfare of any such allocation is at most 
$
(n-\floor{\sqrt{n}}) \cdot \frac{1}{n} + \floor{\sqrt{n}} \cdot \frac{1}{\floor{\sqrt{n}}} \le 2
$, where the second term in the summation is due to the fact that each of the remaining (at most) $\floor{\sqrt{n}}$ goods which may be assigned to high agents can contribute at most $1/\floor{\sqrt{n}}$ to the social welfare. In contrast, the social welfare generated by assigning each high agent the $\floor{\sqrt{n}}$ distinct goods she values positively is at least $\floor{\sqrt{n}} \cdot 1 = \floor{\sqrt{n}}$. Hence, the price of \hmms is at least $\floor{\sqrt{n}}/2 = \Omega(\sqrt{n})$. 
\end{proof}

%% file: discussion.tex
\section{Discussion}\label{sec:disc}
In this work, we focus on the allocation of indivisible goods to $n$ agents. We show that the price of EF1 is $O(\sqrt{n})$ when agent valuations are subadditive (and $\Omega(\sqrt{n})$ even if they are additive), whereas the price of $\nicefrac{1}{2}$-$\MMS$ is $\Theta(\sqrt{n})$ when agent valuations are additive. 

An immediate future direction is to analyze the price of $\nicefrac{1}{2}$-$\MMS$ under more general valuation classes. But even for additive valuations, it would be interesting to consider other (combinations of) fairness and efficiency guarantees. The following are just a few interesting examples. 
\begin{enumerate}
	\item EF1+PO (Pareto optimality), and EF1+fPO (fractional Pareto optimality): fPO implies PO, and both of these allocations are guaranteed to exist~\cite{barman2018finding,CKMP+19}.
	\item $\alpha$-$\MMS$ for $\alpha > \nicefrac{1}{2}$: Even $\nicefrac{3}{4}$-$\MMS$ allocations are guaranteed to exist~\cite{ghodsi2018fair,garg2020improved}.
	\item $\alpha$-GMMS (groupwise maximin share): GMMS implies MMS, and $\nicefrac{1}{2}$-GMMS allocations are guaranteed to exist~\cite{barman2018groupwise}.
	\item EFX (envy-freeness up to any good): EFX is a strengthening of EF1~\cite{CKMP+19}. While it is an open question whether EFX allocations always exist with more than three agents (see~\cite{ChaudhuryGM20} for the case of three agents), one can still define its price by excluding instances (if any) where such allocations do not exist. 
\end{enumerate}

Another direction would be to analyze the price of fairness when the items are \emph{chores}~\cite{aziz2017algorithms}, or a mixture of goods and chores~\cite{aziz2019fair}. More broadly, while formal definitions of fairness and their corresponding prices are relatively well-understood within fair division theory, they are much less explored within other paradigms of social choice theory such as voting or matching. Developing a comprehensive understanding of what fairness means and what its price is requires significant future work.

%% file: lbounds.tex
\section{Price of Fairness with Supermodular Valuations}\label{sec:supermodular}

Let us consider two more classes of agent valuations.
\begin{itemize}
	\item Superadditive: $v_i(S \cup T) \ge v_i(S) + v_i(T)$ for each agent $i \in [n]$ and all disjoint subsets $S,T \subseteq [m]$.
	\item Supermodular: $v_i(S \cup T) + v_i(S \cap T) \ge v_i(S) + v_i(T)$ for each agent $i \in [n]$ and all subsets $S,T \subseteq [m]$.
\end{itemize}

Clearly, if agent valuations are supermodular, then they are also superadditive. Given that we establish the price of $\EFone$ for subadditive valuations (\Cref{theorem:efone}), one might wonder what the price of $\EFone$ is for the complementary class of superadditive valuations. We show that even if the valuations are supermodular and identical across agents, the price of $\EFone$ is unbounded in terms of $n$ and $m$. 

\begin{example}
	Suppose there are $n$ agents with the following identical valuation function $v$ over $m=n$ goods. For all $S \subseteq [m]$, 
	\[
	v(S) = \begin{cases}
	\epsilon\, |S|, &\text{if $|S| \le 1$},\\
	\epsilon + (|S|-1) \cdot L, &\text{if $|S| \ge 1$}.\\
	\end{cases}
	\]
	Here, $\epsilon$ is very small, and $L  = (1-\epsilon)/(m-1)$. It can be checked that the valuation function is supermodular, and that $v([m]) =1$ (and hence valuations are scaled).
	
	The optimal allocation assigns all goods to a single agent, for a social welfare of $1$. However in an $\EFone$ allocation, each agent must receive one good, and hence the maximum social welfare that can be achieved by any $\EFone$ allocation is $n \epsilon$. Hence, the price of $\EFone$ is at least $1/(n \epsilon)$. Since $\epsilon$ can be arbitrarily small, the price of $\EFone$ is unbounded in terms of $n$ and $m$ even for identical supermodular valuations.
\end{example}

\section{Lower Bounds for the Price of Proportionality Up To One Good}

In this section, we show how to modify the construction from the proof of \Cref{theorem:efone} and that due to \citet{BeiLMS19} to derive that the price of Prop1 is $\Omega(n)$ for unscaled additive valuations and $\Omega(\sqrt{n})$ for scaled additive valuations. In each case, the key modification is to have $n+1$ goods instead of $n$ goods. 

\begin{example}[Unscaled Additive Valuations]
	Consider an instance with $n$ agents and $m=n+1$ goods. Suppose the agents have the following additive valuations over the goods: the first agent has value $n+1$ for each good and every other agent has value $1/(n+1)$ for each good. The optimal social welfare is achieved by allocating all goods to the first agent, yielding $\OPT = (n+1)^2$. In contrast, it is easy to see that any Prop1 allocation must give at least one good to each of agents $2,\ldots,n$. Hence, any Prop1 allocation has welfare less than $2n+3$. This implies the desired $\Omega(n)$ bound on the price of Prop1 under unscaled additive valuations. 
\end{example}

\begin{example}[Scaled Additive Valuations]
	Consider an instance with $n$ agents and $m=n+1$ goods. For the ease of exposition, suppose that $\sqrt{n}$ is an integer. For $k \in [\sqrt{n}]$, define $T_k = \set{(k-1)\sqrt{n}+1,\ldots,k\sqrt{n}}$. Note that $T_1,\ldots,T_{\sqrt{n}}$ partition goods $1,\ldots,n$. 
	
	Suppose the agents have the following scaled additive valuations. 
	\begin{itemize}
		\item For $i \le \sqrt{n}$, agent $i$ has value $1/\sqrt{n}$ for each of $\sqrt{n}$ goods in $T_i$, and value $0$ for every other good.
		\item For $i > \sqrt{n}$, agent $i$ has value $1/(n+1)$ for each of $n+1$ goods.  
	\end{itemize}
	
	The optimal social welfare is achieved by partitioning the first $n$ goods among the first $\sqrt{n}$ agents, and giving the last good to a later agent, yielding $\OPT = \sqrt{n}+\frac{1}{n+1} = \Omega(\sqrt{n})$. 
	
	In contrast, any Prop1 allocation must give agent $i$ at least one good for each $i > \sqrt{n}$. Thus, any Prop1 allocation has welfare at most $(\sqrt{n}+1) \cdot \frac{1}{\sqrt{n}} + (n-\sqrt{n}) \cdot \frac{1}{n+1} = O(1)$. This implies the desired $\Omega(\sqrt{n})$ bound on the price of Prop1 under scaled additive valuations. 
\end{example}